\documentclass[%
 aip,
 amsmath,amssymb,
reprint, onecolumn 
]{revtex4-1}

\usepackage{graphicx}
\usepackage{dcolumn}
\usepackage{bm}

\usepackage[utf8]{inputenc}
\usepackage[T1]{fontenc}
\usepackage{mathptmx}
\usepackage{etoolbox}

\usepackage{amsfonts,amssymb,amsthm}
\def\diag{{\rm diag}}

\def\SU{{\rm SU}}
\def\U{{\rm U}}
\def\su{\mathfrak{su}}
\def\C{{\mathbb C}}
\def\R{{\mathbb R}}
\def\SO{{\rm SO}}
\def\OO{{\rm O}}
\def\Spin{{\rm Spin}}
\def\tr{{\rm tr}}
\def\Re{{\rm Re}}

\def\rank{{\rm rank}}

\def\Zero{0}
\newtheorem{theorem}{Theorem}
\newtheorem{remark}{Remark}
\newtheorem{proposition}{Proposition}

\makeatletter
\def\@email#1#2{%
 \endgroup
 \patchcmd{\titleblock@produce}
  {\frontmatter@RRAPformat}
  {\frontmatter@RRAPformat{\produce@RRAP{*#1\href{mailto:#2}{#2}}}\frontmatter@RRAPformat}
  {}{}
}%
\makeatother
\begin{document}

\preprint{AIP/123-QED}

\title[Constant solutions of the SU(2) Yang–Mills–Dirac equations: classification and explicit form]{Constant solutions of the SU(2) Yang–Mills–Dirac equations: classification and explicit form}
\author{D. Shirokov}
 \email{dshirokov@hse.ru}
\affiliation{ 
HSE University, Moscow, Russia
}%

\affiliation{%
Institute for Information Transmission Problems of Russian Academy of Sciences, Moscow, Russia
}%


\begin{abstract}
We present a complete classification of all constant solutions of the
$\SU(2)$ Yang--Mills--Dirac equations in Minkowski space $\R^{1,3}$. The method combines the hyperbolic singular value decomposition of the real
matrix of the gauge potential with the two-sheeted covering of SO(3) by
SU(2). In the degenerate case of zero potential, we classify the solutions
by the causal structure of the current. For nondegenerate solutions with $\Psi\neq 0$ and
$A\neq 0$, the equations force a sharp dichotomy: a nonzero constant
current requires a nonzero mass, and massless solutions necessarily have
zero current. The massive branch gives a constant non-Abelian field
configuration with nonvanishing invariant
$F^2=F_{\mu\nu}F^{\mu\nu}$. The remaining (degenerate) solutions are
classified for completeness. For each case we give explicit formulas for
the spinor $\Psi$, the potential $A$, the current $J$, and the invariant
$F^2$, and we summarize the classification in a table. Nonconstant
solutions can be sought as perturbation-theory series around constant
solutions. We write the linear equations for the first-order correction,
which reduce to linear partial differential equations with constant
coefficients, and find explicit plane-wave solutions of these equations
in the massless case.
\end{abstract}

\maketitle

\section{Introduction}
\label{sect1}

At present, the laws of particle physics are described by gauge theories. The
Yang--Mills theory with the gauge Lie group $\U(1)\times\SU(2)$ describes
electroweak interactions, while with the gauge non-Abelian Lie group
$\SU(3)$ it describes strong interactions. Exact solutions of the classical
Yang--Mills equations are important for the development of gauge theory, in
particular, for describing the vacuum structure of the theory and for a more
complete understanding of the gauge theory. Another possible approach is to
quantize the gauge theory around a classical solution of the Yang--Mills
equations. For further motivation to consider the classical Yang--Mills
theory and its solutions, see the review \cite{MalecRev}. Note that second
quantization issues are not considered in this paper. For the quantum
meaning of classical field theory, see \cite{Jack}.

Constant solutions of the Yang--Mills--Dirac equations are of particular
interest for several reasons. First, they provide exact backgrounds for a
perturbative expansion of the theory: as shown in Section~\ref{sectNS}, one
can use any constant solution as the zeroth approximation and obtain a
consistent perturbation scheme at higher orders. Second, they are natural
candidates for describing the vacuum structure of the theory, in the same
spirit as constant Yang--Mills fields discussed in
\cite{C1, C2, C3, C4, Sch, Sch2}. Third, the instability of constant
Yang--Mills fields \cite{C2, C3} suggests that a similar stability analysis
for the Yang--Mills--Dirac system is a natural question, which we leave for
future work. Finally, a complete classification of the constant solutions
in Minkowski space is a necessary first step towards a local classification
of all (nonconstant) solutions of the classical Yang--Mills--Dirac
equations.

The complexity of studying the Yang--Mills equations is due to the
nonlinearity of these equations: they are cubic with respect to the
potential $A$. In the regime of a small coupling constant, these equations
are solved approximately by perturbation theory. There are many works on
particular classes of solutions of the Yang--Mills equations with zero
current (see, e.g., the classical papers \cite{deA, ADHM, Bel, tH, Pol,
Wit, WYa}, the review \cite{Actor}, and references therein). Constant
solutions of the Yang--Mills equations and their significance for describing
the physical vacuum are discussed in \cite{C1, C2, C3, C4, Sch, Sch2}. We
discuss some particular classes of solutions of the Yang--Mills and
Yang--Mills--Proca equations in \cite{YMP, YMSh, waves, hforms, CC2022,
Shirokov1, Shirokov2}.

The Dirac equation in an external field is linear with respect to the spinor
$\Psi$. Its well-known solutions include, for example, plane waves. The
Dirac equation with respect to the variables $(\Psi, A)$ is nonlinear
because it contains products of the spinor $\Psi$ and the components of the
potential $A$. The right-hand side of the Yang--Mills--Dirac equations
contains the current, which is also nonlinear (quadratic) with respect to
the spinor $\Psi$. The Yang--Mills--Dirac equations form a highly nontrivial
nonlinear system of partial differential equations, too complex to be solved
exactly in the general case. Some particular solutions of this system are
known; see, e.g., the papers \cite{AW, Ant, Ant2, Basler, Magg, Meetz, RTV,
Malec} and the review \cite{MalecRev}.

In this paper, we present a complete classification and the explicit form of
all constant solutions of the Yang--Mills--Dirac equations with $\SU(2)$
gauge symmetry in Minkowski space $\R^{1,3}$. Note that the constant
solutions of the Yang--Mills--Dirac equations are genuinely nonlinear
solutions and, from this point of view, are especially important for
applications. Theorems~\ref{thDYM2} and~\ref{thJzero} are new. Nonconstant solutions can be considered as perturbation series using all
constant solutions as a zeroth approximation. For the first approximation,
the problem reduces to solving systems of linear partial differential
equations with constant coefficients, which can be studied by various
methods of the theory of linear partial differential equations and by
numerical analysis. We obtain explicit plane-wave solutions of these
equations in the massless case. The massive case and the second and
subsequent approximations, which reduce to linear partial differential
equations with variable coefficients, are left for future work. These facts
can be used in the future for the local classification of all (nonconstant)
solutions of the classical Yang--Mills--Dirac equations.

In this paper, we use the method of hyperbolic singular value decomposition
(HSVD). The method of ordinary singular value
decomposition (SVD \cite{SVD1, SVD2}) is standard and is widely used in
various applications. The polar decomposition is related to the SVD and is
used in the Yang--Mills context in \cite{Sim}. The method of hyperbolic
singular value decomposition was proposed in 1991 \cite{Bojan, Bojan2} and
is used in signal and image processing, engineering, and computer science.
A new version of the HSVD that does not use hyperexchange matrices is
presented in \cite{Shirokovhsvd}. It naturally includes the ordinary SVD.
We use the HSVD for the real matrix $A\in\R^{4\times 3}$ of the
Yang--Mills potential to obtain all constant solutions of the
Yang--Mills--Dirac equations.

In Section~\ref{sectYMD}, we discuss the Yang--Mills--Dirac equations. In
Section~\ref{sectCS}, we present the classification and the explicit form
of all constant solutions of the Yang--Mills--Dirac equations with $\SU(2)$
gauge symmetry in Minkowski space $\R^{1,3}$. We give explicit formulas for
the spinor $\Psi$, the potential $A$, the current $J$, and the invariant
$F^2$ defined in~(\ref{F2def}), whose trace enters the Lagrangian of the
Yang--Mills field. In Section~\ref{sectNS}, we discuss the possibility of
studying nonconstant solutions of the Yang--Mills--Dirac equations as
perturbation series using constant solutions as a zeroth approximation, and
we write down the equations for the first approximation. The conclusions
follow in Section~\ref{sectConcl}. In Appendix~\ref{secA}, we present detailed proof of
Theorem~\ref{thDYM2} from
Section~\ref{sectCS}.

\section{The Yang--Mills--Dirac equations}
\label{sectYMD}

Let us consider the Minkowski space $\R^{1,3}$ with Cartesian coordinates $x^\mu$, $\mu=0, 1, 2, 3$. We denote partial derivatives by $\partial_\mu=\frac{\partial}{\partial x^\mu}$. The metric tensor of $\R^{1,3}$ is given by the diagonal matrix
\begin{eqnarray}
\eta=(\eta_{\mu\nu})=(\eta^{\mu\nu})=\diag(1, -1, -1, -1).\label{eta}
\end{eqnarray}
We can raise and lower indices of components of tensor fields with the aid of the metric tensor. For example, $F^{\mu\nu}=\eta^{\mu\alpha} \eta^{\nu\beta}F_{\alpha \beta}$.

Let us consider the Lie group
\begin{eqnarray}
\SU(2)=\{S\in\C^{2\times 2}:\quad  S^\dagger S =I_2,\quad \det (S)=1\}
\end{eqnarray}
and the corresponding Lie algebra
\begin{eqnarray}
\su(2)=\{S\in\C^{2\times 2}:\quad S^\dagger=-S,\quad \tr (S)=0\}.
\end{eqnarray}
Consider the $\SU(2)$ Yang--Mills--Dirac equations
\begin{eqnarray}
&&\partial_\mu A_\nu-\partial_\nu A_\mu-[A_\mu, A_\nu]=:F_{\mu\nu},\label{YM1}\\
&&\partial_\mu F^{\mu\nu} -[A_\mu, F^{\mu\nu}]=J^\nu:=i \Psi^\dagger \gamma^0 \gamma^\nu \Psi-\frac{1}{2}\tr(i \Psi^\dagger \gamma^0 \gamma^\nu\Psi)I_2,\label{YM2}\\
&&i\gamma^\mu (\partial_\mu \Psi+\Psi A_\mu)-m \Psi=0,\qquad m\geq 0,\label{YM3}
\end{eqnarray}
for the unknowns $\Psi: \R^{1,3}\to \C^{4\times 2}$ and $A^\mu: \R^{1,3} \to \su(2)$. The vector field with components $A^\mu$ is the potential of the Yang--Mills field, and the tensor field with components $F_{\mu\nu}=-F_{\nu\mu}: \R^{1,3} \to \su(2)$ is the strength of the Yang--Mills field; $m$ is the mass. 

Note that (\ref{YM1}) can be considered as a definition of the strength $F_{\mu\nu}$. We can substitute (\ref{YM1}) into (\ref{YM2}) and obtain
\begin{eqnarray}
\partial_\mu(\partial^\mu A^\nu-\partial^\nu A^\mu- [A^\mu,A^\nu])- [A_\mu,\partial^\mu A^\nu-\partial^\nu A^\mu-[A^\mu,A^\nu]]=J^\nu.\label{YM4}
\end{eqnarray}
From the system (\ref{YM1}), (\ref{YM2}), (\ref{YM3}), it follows that the current $J^\nu: \R^{1, 3}\to \su(2)$ satisfies the following non-Abelian conservation law
\begin{eqnarray}
\partial _\nu J^\nu-[A_\nu, J^\nu]=0.
\end{eqnarray}
We use the following basis of the Lie algebra $\su(2)$:
\begin{eqnarray}
\tau^1=\frac{1}{2i}\sigma^3,\qquad \tau^2=\frac{1}{2i}\sigma^1,\qquad \tau^3=\frac{1}{2i}\sigma^2,
\end{eqnarray}
where $\sigma^\mu$ are the Pauli matrices
\begin{eqnarray}
\sigma^1=\left( \begin{array}{cc} 0 & 1\\ 1 & 0 \end{array}\right),\quad
\sigma^2=\left( \begin{array}{cc} 0 & -i\\i & 0 \end{array}\right),\quad
\sigma^3=\left( \begin{array}{cc} 1 & 0\\ 0 & -1 \end{array}\right).
\end{eqnarray}
We use the standard Dirac gamma-matrices in (\ref{YM2}) and (\ref{YM3}):
\begin{eqnarray}\gamma^0=\left(
             \begin{array}{cc}
               I_2 & 0 \\
               0 & -I_2 \\
             \end{array}
           \right),\qquad
\gamma^\mu=\left(
               \begin{array}{cc}
                 0 & \sigma^\mu \\
                 -\sigma^\mu & 0 \\
               \end{array}
             \right),\qquad \mu=1, 2, 3.\label{Dgm}
\end{eqnarray}
The system (\ref{YM1}), (\ref{YM2}), (\ref{YM3}) is gauge invariant with respect to the following transformation
\begin{eqnarray}
A_\mu \to  S^{-1}A_\mu S-S^{-1}\partial_\mu S,\qquad F_{\mu\nu}\to S^{-1}F_{\mu\nu}S,\qquad
\Psi \to \Psi S,\qquad J^\mu \to S^{-1} J^\mu S,\qquad S:\R^{1,3}\to\SU(2).\label{gauge}
\end{eqnarray}

In addition to the gauge invariance (\ref{gauge}), the system possesses a
global $U(1)$ phase symmetry: for any constant $\theta\in\mathbb R$, the
transformation
\begin{eqnarray}
\Psi \to e^{i\theta}\Psi\label{globalU1}
\end{eqnarray}
leaves $A_\mu$, $F_{\mu\nu}$, and the current $J^\nu$ unchanged (the
current is quadratic in $\Psi$) and leaves the Dirac equation
(\ref{YM3}) invariant.

We use the notation
\begin{eqnarray}
F^2:=F_{\mu\nu}F^{\mu\nu}=\eta^{\mu\alpha}\eta^{\nu\beta}F_{\mu\nu}F_{\alpha\beta}.
\label{F2def}
\end{eqnarray}
Since $F_{\mu\nu}\in\su(2)$, $F^2$ is a $2\times 2$ matrix; under the gauge
transformation (\ref{gauge}) it transforms as $F^2\to S^{-1}F^2S$. The
gauge-invariant scalar quantity is therefore $\tr(F^2)$, which enters the
Lagrangian of the Yang--Mills field. For all constant solutions constructed
below, $F^2$ is proportional to $I_2$, so the coefficient in $F^2=cI_2$ is
itself gauge invariant; this coefficient is what we quote as $F^2$ in the
classification.

The potential of the Yang--Mills field and the current can be represented in the form
\begin{eqnarray}
A^\mu=A^\mu_{\,\,a}\tau^a,\qquad J^\mu=J^\mu_{\,\,a}\tau^a,\qquad A^\mu_{\,\,a}, J^\mu_{\,\,a}:\R^{1,3}\to \R.
\end{eqnarray}
We denote the corresponding matrices of dimension $4\times 3$ by $A:=(A^\mu_{\,\,a})$ and $J:=(J^\mu_{\,\, a})$ and call them the matrix of potential and the matrix of current.

\section{Classification and explicit form of all constant solutions of the Yang--Mills--Dirac equations}
\label{sectCS}

\subsection{Symmetries of the constant system}
\label{ssecSym}

Let us consider the system for constant solutions (the matrices $\Psi\in\C^{4\times 2}$, $A\in\R^{4\times 3}$ do not depend on $x\in\R^{1,3}$) of the system (\ref{YM1}), (\ref{YM2}), (\ref{YM3}):
\begin{eqnarray}
&&[A_\mu, [A^\mu, A^\nu] ]=J^\nu:=i \Psi^\dagger \gamma^0 \gamma^\nu \Psi-\frac{1}{2}\tr(i \Psi^\dagger \gamma^0 \gamma^\nu\Psi)I_2,\label{fir}\\
&&i \gamma^\mu \Psi A_\mu -m\Psi=0,\qquad m\geq 0.\label{sec}
\end{eqnarray}
For the strength, we get $F^{\mu\nu}=-[A^\mu, A^\nu]$.

In this paper, we present explicit form of all solutions $(\Psi,  A)$ of the system of equations (\ref{fir}), (\ref{sec}). Also we give the explicit formulas for the corresponding current $J=(J^\mu_{\,\,a})$ and the invariant $F^2$ defined in (\ref{F2def}).

The system of equations (\ref{fir}), (\ref{sec}) is invariant under the transformation
\begin{eqnarray}
A^\mu \to  S^{-1}A^\mu S,\qquad \Psi \to \Psi S,\qquad J^\mu \to S^{-1} J^\mu S,\qquad S\in\SU(2),\label{gauge2}
\end{eqnarray}
because of the gauge invariance (\ref{gauge}). The element $S\in\SU(2)$ does not depend on $x$ now. Using the two-sheeted cover of the orthogonal group $\SO(3)$ by the spin group $\Spin(3)\cong\SU(2)$
\begin{eqnarray}
S^{-1}\tau^a S=p^a_b \tau^b,\qquad P=(p^a_b)\in\SO(3),\qquad \pm S\in\SU(2),\label{cover}
\end{eqnarray}
we conclude that the system (\ref{fir}), (\ref{sec}) is invariant under the transformation
\begin{eqnarray}
A \to  AP,\qquad \Psi \to \Psi S,\qquad J\to JP,\qquad S\in\SU(2),\qquad P\in\SO(3),\label{tr11}
\end{eqnarray}
where $P$ and $S$ are related as (\ref{cover}).

The systems of equations (\ref{YM1}), (\ref{YM2}), (\ref{YM3}) and (\ref{fir}), (\ref{sec}) are invariant under the Lorentz transformation of coordinates. Namely, let us consider the transformation $x^\mu \to q^\mu_\nu x^\nu$, $Q=(q^\mu_\nu)\in\OO(1,3)$. The system (\ref{fir}), (\ref{sec}) is invariant under the transformation
\begin{eqnarray}
A \to QA,\qquad J \to QJ,\qquad \Psi \to T \Psi,\qquad Q\in\OO(1,3),\qquad T\in{\rm Pin}(1,3)\,\label{lorentz}
\end{eqnarray}
where $T$ and $Q$ are related as the two-sheeted cover
\begin{eqnarray}
T^{-1}\gamma^\mu T=q^\mu_\nu \gamma^\nu,\qquad Q=(q^\mu_\nu)\in\OO(1,3),\qquad \pm T\in{\rm Pin}(1,3),\label{cover2}
\end{eqnarray}
and $\gamma^\mu$, $\mu=0, 1, 2, 3$, are the Dirac gamma-matrices (\ref{Dgm}).

Combining transformations (\ref{tr11}) and (\ref{lorentz}), we conclude that the system (\ref{fir}), (\ref{sec}) is invariant under the transformation
\begin{eqnarray}
\Psi \to T\Psi S,\qquad A \to Q A P,\qquad J\to QJP,\qquad
S\in\SU(2),\qquad Q\in\OO(1,3),\qquad P\in\SO(3),\qquad T\in{\rm Pin}(1,3),\label{trans2}
\end{eqnarray}
where $P$ and $S$ are related as (\ref{cover}) and $Q$ and $T$ are related as (\ref{cover2}).

\subsection{Nondegenerate solutions with $\Psi\neq0$, $A\neq0$}
\label{ssecNondeg}

We now turn to the nondegenerate case $\Psi\neq 0$, $A\neq 0$, which is the
most interesting from the physical point of view. The degenerate cases
$\Psi=0$ and $A=0$ are treated in Sections~\ref{ssecPsiZero}
and~\ref{ssecAZero}, respectively, and the complete classification is
summarized in Section~\ref{ssecSummary}.

Suppose we have the solution $(\Psi, A)$ of the system (\ref{fir}), (\ref{sec}). We can always choose the matrices $P\in\SO(3)$ and $Q\in\OO(1,3)$ such that $QAP$ is in the canonical form (\ref{DD2}). We use the following theorem on the hyperbolic singular value decomposition (see \cite{Shirokovhsvd}).

\begin{theorem}\label{thHSVD} For an arbitrary matrix $A\in\R^{n\times N}$, there exist
matrices $R\in\OO(N)$ and $L\in\OO(p,q)$, $p+q=n$,
such that
\begin{eqnarray}
L^\top A R=\Sigma^A,\qquad
\Sigma^A= \left(
\begin{array}{cccc}
      I_d &      \Zero & \Zero & \Zero \\
   \Zero &     X_{x} & \Zero & \Zero \\
    \Zero &    \Zero & \Zero & \Zero \\
        \hline
   I_d &      \Zero & \Zero & \Zero \\
  \Zero &      \Zero & Y_y & \Zero \\
  \Zero &      \Zero & \Zero & \Zero\\
      \end{array}
    \right)\begin{array}{c}
      \left.
      \begin{array}{c}
          \\
         \\
          \\
      \end{array}
    \right\}{p} \\
      \left.
      \begin{array}{c}
          \\
         \\
          \\
      \end{array}
    \right\}{q}
    \end{array}\in\R^{n\times N}.\label{DD2}
\end{eqnarray}
Here and below, blocks of zero dimension are allowed. The braces on the right indicate the total number of rows in each of the two super-blocks ($p$ and $q$, respectively), so that $\Sigma^A\in\R^{n\times N}$. The first block of the matrix $\Sigma^A$ has $p$ rows and the second block has $q$ rows. $X_{x}$ and $Y_y$ are diagonal matrices of corresponding dimensions $x$ and $y$ with all positive uniquely determined diagonal elements (up to a permutation). $I_d$ is the identity matrix of dimension $d$. We denote zero blocks by $\Zero$.

Moreover, choosing $R$, one can swap columns of the matrix $\Sigma^A$. Choosing $L$, one can swap rows in individual blocks but not across blocks. Thus we can always arrange diagonal elements of the matrices $X_{x}$ and $Y_y$ in decreasing order.

Here we have
\begin{eqnarray}
d=\rank (A)-\rank (A^\top\eta A),\qquad x+y=\rank (A^\top\eta A),
\end{eqnarray}
$x$ is the number of positive eigenvalues of the matrix $A^\top\eta A$, $y$ is the number of negative eigenvalues of the matrix $A^\top\eta A$.
\end{theorem}

In other words, the HSVD is a generalization of the ordinary SVD to the
pseudo-Euclidean setting. In the case of a real matrix
$A\in\R^{(p+q)\times N}$ and the metric $\eta=\diag(I_p,-I_q)$, the HSVD
reduces $A$ to the canonical form $\Sigma^A$ by means of pseudo-orthogonal
row transformations $L\in\OO(p,q)$ and orthogonal column transformations
$R\in\OO(N)$. The integers $d_A$, $x_A$, $y_A$ of $\Sigma^A$ are invariant
under these transformations: they describe, respectively, the ``rigid''
part of $A$ (the block $I_{d_A}$ that appears in both signatures) and the
numbers of positive and negative hyperbolic singular values. The positive
diagonal entries of $X_{x_A}$ and $Y_{y_A}$ are the hyperbolic singular
values of $A$. The ordinary SVD is recovered in the case $q=0$.

Let us use Theorem  \ref{thHSVD} in the case $p=1$, $q=3$, $n=p+q=4$, $N=3$.
Note that the matrix $\Sigma^A$ has 4 rows and 3 columns in this case. This means that some blocks of the matrix $\Sigma^A$ have size zero.

Note that we can always find the matrix $R=P\in\SO(3)$ from the special orthogonal group in the HSVD. If it has the determinant $-1$, then we can change the sign of all elements of the matrices $P$ and $Q$.

The canonical form of the matrix $A\neq 0$ is characterized by the three parameters $d_A$, $x_A$, and $y_A$. The values of these parameters in all the nine possible cases are presented in Table \ref{tab1}.

\begin{table}[ht]
\caption{\label{tab1}The nondegenerate solutions with $\Psi\neq 0$, $A\neq 0$.
A dash indicates that no solution of this type exists (the corresponding
parameters are not defined).}
\begin{ruledtabular}
\begin{tabular}{ccccccc}
case & $d_A$ & $x_A$ & $y_A$ & $d_J$ & $x_J$ & $y_J$  \\  
\hline
1 &1 & 0 & 0 & 0 & 0 & 0\\ 
2 &1 & 0 & 1 & - & - & - \\ 
3 &1 & 0 & 2 & - & - & - \\ 
4 &0 & 1 & 0 & - & - & - \\ 
5 &0 & 1 & 1 & - & - & - \\ 
6 &0 & 1 & 2 & 0 & 1 & 0 \\ 
7 &0 & 0 & 1 & - & - & - \\ 
8 &0 & 0 & 2 & - & - & - \\ 
9 &0 & 0 & 3 & - & - & -\\ 
\end{tabular}
\end{ruledtabular}
\end{table}

Only Cases~1 and 6 from Table~\ref{tab1} are realized for $\Psi\neq0$ and
$A\neq0$: case~1 corresponds to the massless solution (\ref{caseiii}) and
case~6 to the massive solution (\ref{casei}) of Theorem~\ref{thDYM2}. In the
remaining cases the Dirac equation either forces $\Psi=0$ or has no solution,
so no nondegenerate solution of the corresponding canonical type exists.

The matrix $A\neq 0$ has the following explicit form respectively in these nine cases:
\begin{eqnarray}
           \left(
                \begin{array}{ccc}
                  1 & 0 & 0\\ \hline
                  1 & 0 & 0\\
                  0 & 0 & 0\\
                  0 & 0 & 0\\
                \end{array}
              \right),\quad
\left(
                \begin{array}{ccc}
                  1 & 0 & 0\\ \hline
                  1 & 0 & 0\\
                  0 & a_1 & 0\\
                  0 & 0 & 0\\
                \end{array}
              \right),\quad
              \left(
                \begin{array}{ccc}
                  1 & 0 & 0\\ \hline
                  1 & 0 & 0\\
                  0 & a_1 & 0\\
                  0 & 0 & a_2\\
                \end{array}
              \right),\quad
              \left(
                \begin{array}{ccc}
                  a_1 & 0 & 0\\ \hline
                  0 & 0 & 0\\
                  0 & 0 & 0\\
                  0 & 0 & 0\\
                \end{array}
              \right),\quad
              \left(
                \begin{array}{ccc}
                  a_1 & 0 & 0\\ \hline
                  0 & a_2 & 0\\
                  0 & 0 & 0\\
                  0 & 0 & 0\\
                \end{array}
              \right),\label{cases}\\
              \left(
                \begin{array}{ccc}
                  a_1 & 0 & 0\\ \hline
                  0 & a_2 & 0\\
                  0 & 0 & a_3\\
                  0 & 0 & 0\\
                \end{array}
              \right),\quad
            \left(
                \begin{array}{ccc}
                  0 & 0 & 0\\ \hline
                  a_1 & 0 & 0\\
                  0 & 0 & 0\\
                  0 & 0 & 0\\
                \end{array}
              \right),\quad
              \left(
                \begin{array}{ccc}
                  0 & 0 & 0\\ \hline
                  a_1 & 0 & 0\\
                  0 & a_2 & 0\\
                  0 & 0 & 0\\
                \end{array}
              \right),\quad
              \left(
                \begin{array}{ccc}
                  0 & 0 & 0\\ \hline
                  a_1 & 0 & 0\\
                  0 & a_2 & 0\\
                  0 & 0 & a_3\\
                \end{array}
              \right)\,
              \in\R^{4\times 3}.\nonumber
\end{eqnarray}
Note that we can always choose the matrices $P$ and $Q$ such that the hyperbolic singular values are positive $a_1, a_2, a_3 >0$.

Using canonical form of the matrix $A$, we calculate the current $J$ in (\ref{fir}). The parameters $d_J$, $x_J$, and $y_J$ are also presented in Table \ref{tab1}. Then using the equations (\ref{fir}) and (\ref{sec}), we find the explicit form of $\Psi$. As stated above, only Cases~1 and~6 survive (see the proof of Theorem~\ref{thDYM2} in Appendix~\ref{secA} for details). We use the following notation for the blocks $\Psi_1\in\C^{2\times 2}$ and $\Psi_2\in\C^{2\times 2}$ and elements of the matrix $\Psi\in\C^{4\times 2}$:
\begin{eqnarray}\Psi=\left(
         \begin{array}{c}
           \Psi_1 \\ \hline
           \Psi_2 \\
         \end{array}
       \right)=\left(
  \begin{array}{cc}
    G & B \\
    C & D \\ \hline
    K & L \\
    M & N \\
  \end{array}
\right).\label{notat}
\end{eqnarray}
For brevity, for $\Phi_1,\Phi_2\in\C^{2\times2}$ we write
$(\Phi_1;\Phi_2)$ for the block column
$\left(\begin{array}{c}\Phi_1\\\hline\Phi_2\end{array}\right)\in\C^{4\times2}$.

Finally, we get Theorem \ref{thDYM2}.

\begin{theorem}\label{thDYM2} All solutions $(\Psi, A)$  with $\Psi\neq 0$, $A\neq 0$ of the system of equations (\ref{fir}), (\ref{sec}) can be reduced using transformations (\ref{trans2}) to the following two types of solutions:
\begin{enumerate}
\item the solutions with nonzero mass and nonzero current
\begin{eqnarray}
\Psi=\left(
                 \begin{array}{cc}
                   0 & 0 \\
                   0 & 0 \\ \hline
                   0 & L \\
                   0 & 0 \\
                 \end{array}
               \right),\,\,\,\,
                A=(A^\mu_{\,\, a})=\left(
                    \begin{array}{ccc}
                      2m & 0 & 0 \\ \hline
                      0 & 2m & 0 \\
                      0 & 0 & 2m \\
                      0 & 0 & 0 \\
                    \end{array}
                  \right),\,\,\,\,
J=(J^\mu_{\,\, a})=\left(
                                       \begin{array}{ccc}
                                         16m^3 & 0 & 0 \\ \hline
                                         0 & 0 & 0 \\
                                         0 & 0 & 0 \\
                                         0 & 0 & 0 \\
                                       \end{array}
                                     \right),\,\,\,\, F_{\mu\nu}\ne0,\,\,\,\, F^2=8m^4 I_2,\,\,\,\, m>0,\label{casei}
\end{eqnarray}
where $L\in\C$ satisfies $|L|=4m^{\frac{3}{2}}$. The phase of $L$ does not
enter $J$ or $F^2$ and is removed by the global $U(1)$ phase symmetry
(\ref{globalU1}). Thus, up to this symmetry, one may take
$L=4m^{3/2}>0$. We denote this branch by $\mathrm{(M)}$.

\item the solutions with zero mass and zero current
\begin{eqnarray}
m=0,\qquad \Psi=\left(
  \begin{array}{cc}
    G & B \\
    C & D \\ \hline
    C & D \\
    G & B \\
  \end{array}
\right),\qquad  A=\left(
                \begin{array}{ccc}
                  1 & 0 & 0\\ \hline
                  1 & 0 & 0\\
                  0 & 0 & 0\\
                  0 & 0 & 0\\
                \end{array}
              \right),\qquad
J=0,\qquad F_{\mu\nu}=0,\qquad F^2=0,\label{caseiii}
\end{eqnarray}
where
\begin{eqnarray}
|G|^2+|C|^2=|B|^2+|D|^2,\qquad \bar{G}B+\bar{C}D=0,\qquad (G, B, C, D)\neq(0,0,0,0),\qquad G, B, C, D\in\C.\label{cond1}
\end{eqnarray}
\end{enumerate}
\end{theorem}
\begin{proof} The detailed proof is presented in Appendix \ref{secA}. \end{proof}

\begin{proposition}\label{remSimplifiedCaseiii}
The solution (\ref{caseiii}) can be simplified further by a suitable
choice of the transformation (\ref{trans2}), at the cost of changing $A$
from its canonical form. Namely, using the residual symmetry one can
reduce $\Psi$ to the fixed matrix
\begin{eqnarray}
\Psi=\psi\left(
  \begin{array}{cc}
    1 & 0 \\
    0 & 1 \\ \hline
    1 & 0 \\
    0 & -1 \\
  \end{array}
\right),\qquad \psi>0,\label{caseiii-simple}
\end{eqnarray}
while $A$ takes the form
\begin{eqnarray}
A^\mu_{\,\, a}=l^\mu p_a,\qquad \mu=0,1,2,3,\quad a=1,2,3,\label{Asimple}
\end{eqnarray}
where $l=(1,l^1,l^2,l^3)\in\R^{1,3}$ is a future-directed null vector
normalized by $l^0=1$, so that
\begin{eqnarray}
(l^1)^2+(l^2)^2+(l^3)^2=1,
\end{eqnarray}
and $p=(p_1,p_2,p_3)\in\R^3$ is a unit vector,
\begin{eqnarray}
p_1^2+p_2^2+p_3^2=1.
\end{eqnarray}
\end{proposition}
\begin{proof}
Let $\Phi=\begin{pmatrix}G&B\\C&D\end{pmatrix}$ so that $\Psi=
(\Phi;\sigma^1\Phi)$ in (\ref{caseiii}). The conditions (\ref{cond1})
mean that $\Phi^\dagger\Phi=|u|^2I_2$, where $u=(G,C)^T\ne0$. Hence
$\Phi=\psi U$ with $\psi>0$ and $U\in U(2)$. Write
$U=e^{i\varphi}U_0$ with $U_0\in\SU(2)$.

\emph{Step 1 (global phase).} The transformation $\Psi\to e^{-i\varphi}\Psi$
belongs to the global $U(1)$ symmetry (\ref{globalU1}) and gives
$\Phi=\psi U_0$.

\emph{Step 2 ($SU(2)$-part).} Taking $S=U_0^{-1}\in\SU(2)$ in
(\ref{trans2}) gives $\Phi=\psi I_2$, i.e.\ $\Psi=\psi(I_2;\sigma^1)$.

\emph{Step 3 (Lorentz part).} Let $A_0\in\SU(2)$ satisfy
$A_0\sigma^1A_0^{-1}=\sigma^3$; for instance
$A_0=\frac{1}{\sqrt2}\begin{pmatrix}1&1\\-1&1\end{pmatrix}$. The matrix
$T=\diag(A_0,A_0)$ lies in ${\rm Pin}(1,3)$ and corresponds to
$Q=\diag(1,R_0)\in\OO(1,3)$ with $R_0\in\SO(3)$ the image of $A_0$ under
the double cover $\SU(2)\to\SO(3)$. Indeed, in the standard
representation (\ref{Dgm}),
\begin{eqnarray*}
T^{-1}\gamma^0T=\gamma^0,\qquad
T^{-1}\gamma^iT=R_0{}^i{}_j\,\gamma^j,
\end{eqnarray*}
where the second relation follows from
$A_0^{-1}\sigma^iA_0=R_0{}^i{}_j\sigma^j$, i.e.\ from the definition of
the cover $\SU(2)\to\SO(3)$. Hence $T^{-1}\gamma^\mu T=q^\mu{}_\nu\gamma^\nu$
with $Q=(q^\mu{}_\nu)=\diag(1,R_0)$. Then
\begin{eqnarray*}
T(I_2;\sigma^1)=(A_0;A_0\sigma^1)=(I_2;\sigma^3)A_0,
\end{eqnarray*}
and absorbing $A_0$ into a further right multiplication by
$S=A_0^{-1}\in\SU(2)$ gives (\ref{caseiii-simple}) with
$\Psi=\psi(I_2;\sigma^3)$. In the same transformation (\ref{trans2})
this right multiplication acts on $A$ through the corresponding
$P_0\in\SO(3)$ (the image of $S=A_0^{-1}$ under the double cover
$\SU(2)\to\SO(3)$), so that $A\mapsto QAP_0$; together with the
right multiplication from Step~2 this gives a total $P\in\SO(3)$.

\emph{Step 4 (transformation of $A$).} The canonical $A$ in (\ref{caseiii}) has the form
$A^\mu_{\,\, a}=k^\mu p^{(0)}_a$ with $k=(1,1,0,0)$ and $p^{(0)}=(1,0,0)$. Under (\ref{trans2}), with $Q=\diag(1,R_0)$ and with the total right multiplier $P\in\SO(3)$ from Steps~2 and~3, we have $A\to QAP$, hence
\begin{eqnarray*}
A^\mu_{\,\, a}=Q^\mu{}_\nu k^\nu\,P^1{}_a=l^\mu p_a,
\end{eqnarray*}
where $l=Qk\in\R^{1,3}$ and $p=(P^1{}_1,P^1{}_2,P^1{}_3)\in\R^3$. In Step~3 the Lorentz transformation was taken as
$Q=\diag(1,R_0)\in\OO(1,3)$ with $R_0\in\SO(3)$, i.e.\ without a boost;
hence $Q^0{}_0=1$ and $Q^0{}_i=0$, so that
\begin{eqnarray*}
l^0=Q^0{}_\nu k^\nu=Q^0{}_0k^0=1.
\end{eqnarray*}
Since $Q\in\OO(1,3)$ preserves $\eta$ and $k$ is null, $l=Qk$ is also
null:
\begin{eqnarray*}
\eta_{\mu\nu}l^\mu l^\nu=\eta_{\mu\nu}k^\mu k^\nu=0.
\end{eqnarray*}
Together with $l^0=1$ this gives
$(l^1)^2+(l^2)^2+(l^3)^2=1$, i.e.\ $l$ is a future-directed null vector
normalized by $l^0=1$. Since $P\in\SO(3)$, its first row $p$
is a unit vector, $p_1^2+p_2^2+p_3^2=1$. This gives (\ref{Asimple}).

Since (\ref{trans2}) is a symmetry of the system, the transformed pair
$(\Psi,A)$ is again a solution. For $F^2=F_{\mu\nu}F^{\mu\nu}$ we use
that $F^2$ is a Lorentz scalar and is proportional to $I_2$, so it is invariant under (\ref{trans2}); since $F^2=0$ for
(\ref{caseiii}), it remains zero. The quantities $J^\nu$ and
$F_{\mu\nu}$ transform covariantly and vanish for (\ref{caseiii}), hence
they also remain zero.
\end{proof}

\begin{remark} Note that the simplified form (\ref{caseiii-simple}) is precisely the
canonical representative $\Psi^{(L)}(\psi)$ of class $\mathrm{(L)}$ in
Theorem~\ref{thJzero}, where $J^\nu=0$ is established.
\end{remark}
\begin{remark}\label{remParamCount}
The parametrization (\ref{caseiii-simple})--(\ref{Asimple}) contains the
same number of real parameters as (\ref{caseiii}). In (\ref{caseiii})
the four complex entries $G,B,C,D$ provide eight real parameters subject
to the conditions (\ref{cond1}), which consist of one real equation and
one complex equation, i.e.\ three real equations. The number of
parameters is therefore $8-3=5$. Equivalently, the conditions
(\ref{cond1}) are equivalent to $\Phi=\psi U$ with $\psi>0$ and
$U\in\U(2)$, where $\Phi=\begin{pmatrix}G&B\\C&D\end{pmatrix}$. Since
$\dim\U(2)=4$, this again gives $1+4=5$ parameters. In the parametrization
(\ref{caseiii-simple})--(\ref{Asimple}) the parameters are $\psi>0$, the
unit vector $(l^1,l^2,l^3)\in\R^3$, and the unit vector $p\in\R^3$,
which gives $1+2+2=5$ parameters. Thus no additional degrees of freedom
are introduced by the simplification.
\end{remark}

Now let us consider the degenerate solutions $(\Psi, A)$ with $A=0$ or $\Psi=0$.

\subsection{Degenerate solutions with $\Psi=0$}
\label{ssecPsiZero}

If $\Psi=0$, then (\ref{sec}) holds for arbitrary $A$, the current is zero $J=0$, and (\ref{fir}) is reduced to
\begin{eqnarray}
[A_\mu, [A^\mu, A^\nu] ]=0.
\end{eqnarray}
In \cite{Shirokov2}, constant solutions of the Yang--Mills equations with an arbitrary constant current $J^\nu$ were classified, i.e.\ the equation
$[A_\mu,[A^\mu,A^\nu]]=J^\nu$ was solved for a given $J^\nu$. In our case
$\Psi=0$ implies $J^\nu=0$, so the equation reduces to the case
$[A_\mu,[A^\mu,A^\nu]]=0$. Taking the special case $J^\nu=0$ of the
classification obtained in \cite{Shirokov2} (which uses the same HSVD
reduction as in Theorem~\ref{thHSVD}), we arrive at Theorem~\ref{thR}. The parameters $d_A$, $x_A$, and $y_A$ for these solutions are presented in Table \ref{tab2}.

\begin{table}[ht]
\caption{The degenerate solutions with $\Psi=0$}\label{tab2}%
\begin{ruledtabular}
\begin{tabular}{ccccccc}
case & $d_A$ & $x_A$ & $y_A$ & $d_J$ & $x_J$ & $y_J$  \\  
\hline
1 &0&0&0&0&0&0\\ 
2 &1&0&0&0&0&0 \\ 
3 &0&1&0&0&0&0 \\ 
4 &0&0&1&0&0&0 \\ 
\end{tabular}
\end{ruledtabular}
\end{table}

\begin{theorem} \label{thR} All solutions $(\Psi, A)$  with $\Psi=0$ of the system of equations (\ref{fir}), (\ref{sec}) can be reduced using transformations (\ref{trans2}) to the following four types of solutions:
\begin{enumerate}
  \item \begin{eqnarray}
 \Psi=0,\qquad A=0,\qquad J=0,\qquad F_{\mu\nu}=0,\qquad F^2=0;\label{triv1}
  \end{eqnarray}
   \item \begin{eqnarray}\Psi=0,\qquad A=\left(
                                              \begin{array}{ccc}
                                                1 & 0 & 0 \\ \hline
                                                1 & 0 & 0 \\
                                                0 & 0 & 0 \\
                                                0 & 0 & 0 \\
                                              \end{array}
                                            \right),\qquad J=0,\qquad  F_{\mu\nu}=0,\qquad F^2=0;\label{triv4}\end{eqnarray}
  \item \begin{eqnarray}\Psi=0,\qquad A=\left(
                                              \begin{array}{ccc}
                                                a & 0 & 0 \\ \hline
                                                0 & 0 & 0 \\
                                                0 & 0 & 0 \\
                                                0 & 0 & 0 \\
                                              \end{array}
                                            \right),\qquad a>0,\qquad J=0,\qquad  F_{\mu\nu}=0,\qquad F^2=0;\label{triv2}\end{eqnarray}
  \item \begin{eqnarray}\Psi=0,\qquad A=\left(
                                              \begin{array}{ccc}
                                                0 & 0 & 0 \\ \hline
                                                a & 0 & 0 \\
                                                0 & 0 & 0 \\
                                                0 & 0 & 0 \\
                                              \end{array}
                                            \right),\qquad a>0,\qquad J=0,\qquad  F_{\mu\nu}=0,\qquad F^2=0.\label{triv3}\end{eqnarray}
\end{enumerate}
\end{theorem}
\begin{proof}
The proof follows the same HSVD reduction as in Theorem~\ref{thHSVD} and is
given in \cite{Shirokov2} for the case of an arbitrary constant current
$J^\nu$. Setting $J^\nu=0$ leaves the four canonical forms listed above.
\end{proof}

\subsection{Degenerate solutions with $A=0$}
\label{ssecAZero}

Now let us consider the degenerate solutions $(\Psi, A)$ with $A=0$.

If $A=0$ and $m\neq 0$, then $J=0$ and $\Psi=0$, i.e. we get the solutions (\ref{triv1}).

If $A=0$ and $m=0$, then $J=0$, but $\Psi$ can be not zero. We get the equation
\begin{eqnarray}
J^\nu:=i \Psi^\dagger \gamma^0 \gamma^\nu \Psi-\frac{1}{2}\tr(i \Psi^\dagger \gamma^0 \gamma^\nu\Psi)I_2=0.\label{currentzero}
\end{eqnarray}
Substituting an arbitrary
\begin{eqnarray}
\Psi=\left(
  \begin{array}{cc}
    G & B \\
    C & D \\
    K & L \\
    M & N \\
  \end{array}
\right),
\end{eqnarray}
we get
\begin{eqnarray}
J^\nu=\frac{i}{2}\left(\begin{array}{cc}X^\nu & 2Y^\nu \\ 2\bar{Y}^\nu &-X^\nu\\ \end{array} \right),
\end{eqnarray}
where
\begin{eqnarray*}
&&X^0=|G|^2+|C|^2+|K|^2+|M|^2-|B|^2-|D|^2-|L|^2-|N|^2,\qquad
Y^0=\bar{G}B+\bar{C}D+\bar{K}L+\bar{M}N,\\
&&X^1=\bar{M}G+\bar{K}C+\bar{C}K+\bar{G}M-\bar{N}B-\bar{L}D-\bar{D}L-\bar{B}N,\qquad
Y^1=\bar{M}B+\bar{K}D+\bar{C}L+\bar{G}N,\\
&&X^2=i(\bar{M}G-\bar{K}C+\bar{C}K-\bar{G}M-\bar{N}B+\bar{L}D-\bar{D}L+\bar{B}N),\qquad
Y^2=i(\bar{M}B-\bar{K}D+\bar{C}L-\bar{G}N),\\
&&X^3=\bar{K}G-\bar{M}C+\bar{G}K-\bar{C}M-\bar{L}B+\bar{N}D-\bar{B}L+\bar{D}N,\qquad
Y^3=\bar{K}B-\bar{M}D+\bar{G}L-\bar{C}N.
\end{eqnarray*}
Using $J^\nu=0$, we get the conditions
\begin{eqnarray}
&&|G|^2+|C|^2+|K|^2+|M|^2=|B|^2+|D|^2+|L|^2+|N|^2,\quad
\bar{M}G+\bar{C}K-\bar{N}B-\bar{D}L=0,\quad \bar{L}G-\bar{N}C+\bar{B}K-\bar{D}M=0,\label{condit}\\
&&\bar{G}B+\bar{C}D+\bar{K}L+\bar{M}N=0,\quad \bar{N}G+\bar{D}K=0,\quad \bar{L}C+\bar{B}M=0,\quad \Re(\bar{K}G-\bar{M}C-\bar{L}B+\bar{N}D)=0,\nonumber
\end{eqnarray}
The system (\ref{condit}) consists of seven conditions (equivalently,
twelve real quadratic equations) in the eight complex variables
$G,B,C,D,K,L,M,N$ (sixteen real variables), and solving it directly is cumbersome.

\begin{theorem}\label{thJzero}
The system (\ref{currentzero}),
up to the transformations (\ref{trans2}) and the global $U(1)$ phase
(\ref{globalU1}), has the following nonzero solutions:
\begin{eqnarray}
\mathrm{(T)}\quad
&\Psi^{(T)}(\lambda,\kappa)=
\kappa\begin{pmatrix}
I_2\\[1mm]
i\lambda I_2
\end{pmatrix},
\qquad \lambda\in\R,\ \kappa>0,\label{tr1}\\[4mm]
\mathrm{(P)}\quad
&\Psi^{(P)}(\kappa)=
\kappa\begin{pmatrix}
0_{2\times2}\\[1mm]
I_2
\end{pmatrix},
\qquad \kappa>0,\label{tr2}\\[4mm]
\mathrm{(L)}\quad
&\Psi^{(L)}(\kappa)=
\kappa\begin{pmatrix}
I_2\\[1mm]
\sigma^3
\end{pmatrix},
\qquad \kappa>0.\label{tr3}
\end{eqnarray}
There are no other solutions.
\end{theorem}
Here the labels refer to the structure of the canonical form:
$(T)$ and $(L)$ denote the timelike and lightlike cases for the vector
$c^\mu$ in the proof, and $(P)$ denotes the case where the upper
$2\times2$ block of $\Psi$ vanishes. In all three cases $\Psi\ne0$.
\begin{remark}\label{remO13redundancy}
The Lorentz group in (\ref{trans2}) is $O(1,3)$, not $SO^+(1,3)$. Under
the full $O(1,3)$, the classes $(T)$ and $(P)$ are not independent: the
matrix $\gamma^5:=i\gamma^0\gamma^1\gamma^2\gamma^3$ lies in
$\mathrm{Pin}(1,3)$ with $Q=-I_4$ and satisfies
$\gamma^5\Psi^{(T)}(0,\kappa)=\Psi^{(P)}(\kappa)$. Thus, modulo $O(1,3)$,
the class $(P)$ is the $\gamma^5$-image of the representative
$\lambda=0$ of class $(T)$; it is listed separately for convenience.
If the Lorentz transformations in (\ref{trans2}) were restricted to
$SO^+(1,3)$, the three classes $(T)$, $(P)$, $(L)$ would be independent.
\end{remark}

\begin{proof}
Assume $\Psi\ne0$. Set
\[
M^\nu=i\Psi^\dagger\gamma^0\gamma^\nu\Psi .
\]
The condition $J^\nu=0$ is
equivalent to
\[
M^\nu=\frac12\tr(M^\nu)I_2 .
\]
Hence
\[
\Psi^\dagger\gamma^0\gamma^\nu\Psi=c^\nu I_2,
\qquad c^\nu\in\R.
\]
For $\nu=0$ we obtain
\[
\Psi^\dagger\Psi=c^0 I_2,
\qquad c^0\ge0.
\]
If $c^0=0$, then $\Psi=0$, contrary to the assumption. Thus $c^0>0$.

For any vector $v\in\C^2$, consider the spinor $\psi=\Psi v$. Then
\[
\psi^\dagger\gamma^0\gamma^\nu\psi
=
c^\nu\,\|v\|^2.
\]
We show that $V^\nu=\psi^\dagger\gamma^0\gamma^\nu\psi$ is
future-directed causal. Write $\psi=(\xi;\eta)$ with $\xi,\eta\in\C^2$.
Then
\[
V^0=|\xi|^2+|\eta|^2\ge0,\qquad
V^i=2\operatorname{Re}\bigl(\xi^\dagger\sigma^i\eta\bigr)\in\R,
\quad i=1,2,3 .
\]
Using direct calculation, we get
\[
\sum_{i=1}^3\bigl(2\operatorname{Re}(\xi^\dagger\sigma^i\eta)\bigr)^2
=4|\xi|^2|\eta|^2-4\bigl(\operatorname{Im}\xi^\dagger\eta\bigr)^2 ,
\]
and hence
\[
\eta_{\mu\nu}V^\mu V^\nu=(V^0)^2-\sum_{i=1}^3(V^i)^2
=(|\xi|^2-|\eta|^2)^2+4\bigl(\operatorname{Im}\xi^\dagger\eta\bigr)^2\ge0 .
\]
Thus $V$, and hence $c^\nu$, is a future-directed causal vector.

By a Lorentz transformation $T\in\Spin(1,3)$ it can be
brought to one of the two canonical forms
\[
c^\mu=(a,0,0,0),\quad a>0,
\qquad\text{or}\qquad
c^\mu=(a,0,0,a),\quad a>0 .
\]
The reduction to the two canonical forms is standard. By a spatial
rotation $R\in SO(3)$ the vector $\mathbf c=(c^1,c^2,c^3)$ is brought to
the form $|\mathbf c|\,e_3$, so that
$(c^0,c^1,c^2,c^3)\mapsto(c^0,0,0,|\mathbf c|)$. If $c$ is timelike,
$(c^0)^2-|\mathbf c|^2>0$, a boost along $e_3$ with
$\tanh\chi=|\mathbf c|/c^0$ brings it to $(a,0,0,0)$ with
$a=\sqrt{(c^0)^2-|\mathbf c|^2}>0$. If $c$ is lightlike,
$|\mathbf c|=c^0$, the same rotation already gives $(a,0,0,a)$ with
$a=c^0>0$. Both transformations are in $SO^+(1,3)$. It remains to solve $J^\nu=0$ in these
two canonical frames with the parameter $a>0$ explicit.

\medskip
\noindent\textit{1. Timelike case.}
Bring $c^\mu$ to the form $(a,0,0,0)$ with $a>0$. In the standard Dirac
representation (\ref{Dgm}) write
\[
\Psi=
\begin{pmatrix}
P\\[1mm]
Q
\end{pmatrix},
\qquad
P,Q\in M_{2\times2}(\C).
\]
The conditions $J^\nu=0$ become
\[
P^\dagger P+Q^\dagger Q=aI_2,
\]
\[
P^\dagger\sigma^i Q+Q^\dagger\sigma^i P=0,
\qquad i=1,2,3.
\]
Let $p_\alpha,q_\alpha$ be the columns of $P$ and $Q$, $\alpha=1,2$. Then
\begin{eqnarray}
p_\alpha^\dagger p_\beta+q_\alpha^\dagger q_\beta=a\delta_{\alpha\beta},
\label{eq:PQnorm}\\
p_\alpha^\dagger\sigma^i q_\beta+q_\alpha^\dagger\sigma^i p_\beta=0.
\label{eq:PQcross}
\end{eqnarray}
For $\alpha=\beta$, the second equality gives
\[
2\operatorname{Re}\bigl(p_\alpha^\dagger\sigma^i q_\alpha\bigr)=0.
\]
If $p_\alpha\ne0$, this forces
\[
q_\alpha=i\lambda_\alpha p_\alpha,
\qquad
\lambda_\alpha\in\R.
\]
Indeed, fix $\alpha$ with $p_\alpha\ne0$ and put
$M_\alpha=q_\alpha p_\alpha^\dagger$. From
$2\operatorname{Re}(p_\alpha^\dagger\sigma^iq_\alpha)=0$ we get
$\operatorname{tr}((M_\alpha+M_\alpha^\dagger)\sigma^i)=0$ for $i=1,2,3$.
The matrix $H=M_\alpha+M_\alpha^\dagger$ is Hermitian, and
$\{I_2,\sigma^1,\sigma^2,\sigma^3\}$ is a basis of the real vector space
of Hermitian $2\times2$ matrices. Writing
$H=cI_2+\sum_i d_i\sigma^i$ with $c,d_i\in\R$, we have
$\operatorname{tr}(H\sigma^i)=2d_i$. Hence $d_i=0$ and
$H=c_\alpha I_2$ with $c_\alpha\in\R$. Multiplying this equality by
$q_\alpha$ on the right yields
\[
|q_\alpha|^2p_\alpha=(c_\alpha-p_\alpha^\dagger q_\alpha)q_\alpha,
\]
so $p_\alpha$ and $q_\alpha$ are linearly dependent over $\C$:
$q_\alpha=\mu_\alpha p_\alpha$ with $\mu_\alpha\in\C$. Substituting into
the conditions gives
$\operatorname{Re}\mu_\alpha\cdot(p_\alpha^\dagger\sigma^ip_\alpha)=0$
for all $i$. Since $p_\alpha^\dagger\sigma^ip_\alpha\in\R$ and
\[
\sum_{i=1}^3(p_\alpha^\dagger\sigma^ip_\alpha)^2=|p_\alpha|^4>0
\]
(indeed, for $p_\alpha=(u,v)^T\in\C^2$ this sum equals
$(|u|^2+|v|^2)^2$), this forces $\operatorname{Re}\mu_\alpha=0$, i.e.\
$\mu_\alpha=i\lambda_\alpha$ with $\lambda_\alpha\in\R$.

If at least one $p_\alpha=0$, the cross conditions and the normalization
conditions imply $P=0$. Indeed, if $p_1=0$, then (\ref{eq:PQnorm}) with $\alpha=\beta=1$ gives $q_1^\dagger q_1=a>0$,
so $q_1\ne0$; (\ref{eq:PQcross}) with $\alpha=1$, $\beta=2$ gives
$q_1^\dagger\sigma^ip_2=0$ for $i=1,2,3$. Since $q_1\ne0$, the vectors
$\sigma^1q_1,\sigma^2q_1,\sigma^3q_1$ span $\C^2$, hence $p_2=0$.
Together with $p_1=0$ this gives $P=0$. The case $p_2=0$ is analogous.

Then
\[
Q^\dagger Q=aI_2,
\]
that is,
\[
\Psi=
\begin{pmatrix}
0\\[1mm]
\sqrt{a}\,U
\end{pmatrix},
\qquad U\in\U(2).
\]
Write $U=e^{i\varphi}U_0$ with $U_0\in\SU(2)$ and
$\varphi\in\R$. Taking $S=U_0^{-1}\in\SU(2)$ and using the global
$U(1)$ transformation (\ref{globalU1}) with $\theta=-\varphi$, we get
\[
e^{i\theta}\Psi S=
\begin{pmatrix}
0\\[1mm]
\sqrt{a}\,I_2
\end{pmatrix}.
\]
Thus, modulo the global $U(1)$ phase, this is class $\mathrm{(P)}$.
Writing $\kappa=\sqrt{a}>0$, we obtain the canonical representative
$\Psi^{(P)}(\kappa)$.

Now suppose $P\ne0$. Then all $p_\alpha\ne0$ and $q_\alpha=i\lambda_\alpha p_\alpha$. The
column orthogonality conditions give
\[
(1+\lambda_1\lambda_2)p_1^\dagger p_2=0,
\qquad
i(\lambda_2-\lambda_1)p_1^\dagger\sigma^i p_2=0.
\]
If $\lambda_1\ne\lambda_2$, then $p_1^\dagger\sigma^i p_2=0$ for all $i$,
which is impossible for nonzero $p_1,p_2$. Hence $\lambda_1=\lambda_2=\lambda$.
Then
\[
P^\dagger P=\frac{a}{1+\lambda^2}I_2,
\]
so $P=\sqrt{a/(1+\lambda^2)}\,U$ with $U\in\U(2)$, and
\[
Q=i\lambda P.
\]
Writing $\kappa=\sqrt{a/(1+\lambda^2)}>0$, we obtain
\[
\Psi=
\kappa\begin{pmatrix}
I_2\\[1mm]
i\lambda I_2
\end{pmatrix}U,
\qquad
\lambda\in\R.
\]
Write $U=e^{i\varphi}U_0$ with $U_0\in\SU(2)$ and
$\varphi\in\R$. Taking $S=U_0^{-1}\in\SU(2)$ and using the global
$U(1)$ transformation (\ref{globalU1}) with $\theta=-\varphi$, we get
\[
e^{i\theta}\Psi S=
\kappa\begin{pmatrix}
I_2\\[1mm]
i\lambda I_2
\end{pmatrix}.
\]
Thus, modulo the global $U(1)$ phase, this is class $\mathrm{(T)}$.

\medskip
\noindent\textit{2. Lightlike case.}
Bring $c^\mu$ to the form $(a,0,0,a)$ with $a>0$. It is convenient to pass to the
Weyl basis
\[
\gamma^0=
\begin{pmatrix}
0&I_2\\
I_2&0
\end{pmatrix},
\qquad
\gamma^\mu=
\begin{pmatrix}
0&\sigma^\mu\\
-\sigma^\mu&0
\end{pmatrix},\qquad \mu=1, 2, 3,
\]
which is related to the standard Dirac representation (\ref{Dgm}) by a
unitary change of basis. Write
\[
\Psi=
\begin{pmatrix}
\xi\\[1mm]
\eta
\end{pmatrix},
\qquad \xi,\eta\in M_{2\times2}(\C).
\]
The conditions $J^\nu=0$ give, in the Weyl basis,
\begin{eqnarray}
&&\xi^\dagger\xi+\eta^\dagger\eta=aI_2,
\label{WL0}\\
&&-\xi^\dagger\sigma^1\xi+\eta^\dagger\sigma^1\eta=0,
\qquad
-\xi^\dagger\sigma^2\xi+\eta^\dagger\sigma^2\eta=0,
\qquad
-\xi^\dagger\sigma^3\xi+\eta^\dagger\sigma^3\eta=aI_2.
\label{WL3}
\end{eqnarray}
Here we used $\gamma^0\gamma^i=\diag(-\sigma^i,\sigma^i)$ in the Weyl
basis, so that
$\Psi^\dagger\gamma^0\gamma^i\Psi=-\xi^\dagger\sigma^i\xi+\eta^\dagger\sigma^i\eta$.

\emph{Step 1 (block structure).} Adding and subtracting (\ref{WL0}) and
the third equation of (\ref{WL3}) gives
\begin{eqnarray*}
\xi^\dagger(I_2-\sigma^3)\xi+\eta^\dagger(I_2+\sigma^3)\eta&=&2aI_2,\\
\xi^\dagger(I_2+\sigma^3)\xi+\eta^\dagger(I_2-\sigma^3)\eta&=&0 .
\end{eqnarray*}
Since $I_2+\sigma^3=\diag(2,0)$ and $I_2-\sigma^3=\diag(0,2)$, taking
the trace of the second matrix equation gives
\[
2\bigl(|\xi_{11}|^2+|\xi_{12}|^2\bigr)
+2\bigl(|\eta_{21}|^2+|\eta_{22}|^2\bigr)=0 .
\]
We get
$\xi_{11}=\xi_{12}=\eta_{21}=\eta_{22}=0$ and
\[
\xi=\begin{pmatrix}0_{1\times2}\\ u^T\end{pmatrix},\qquad
\eta=\begin{pmatrix}v^T\\ 0_{1\times2}\end{pmatrix},\qquad
u,v\in\C^2 .
\]
The first equation then reduces to
\begin{eqnarray}
uu^\dagger+vv^\dagger=aI_2.
\label{WLn}
\end{eqnarray}
The first two equations in (\ref{WL3}) are satisfied identically,
because $\sigma^1_{11}=\sigma^1_{22}=\sigma^2_{11}=\sigma^2_{22}=0$ and
the nonzero rows of $\xi,\eta$ are the second and first, respectively.

\emph{Step 2 (right $\U(2)$ rotation modulo the global $\U(1)$ phase).} By (\ref{WLn}), the columns of
the $2\times2$ matrix $(u,v)$ are orthonormal up to the factor
$\sqrt{a}$. Hence $u=\sqrt{a}\,Ue_1$, $v=\sqrt{a}\,Ue_2$ for some
$U\in\U(2)$. Set
\[
R:=(e_2,e_1)=\begin{pmatrix}0&1\\1&0\end{pmatrix},\qquad
S:=\bar U R\in\U(2),
\]
where $\bar U$ is the entrywise complex conjugate of $U$ (so that
$U^T\bar U=I$). The multiplier $S\in\U(2)$ is admissible modulo the
global $U(1)$ phase (\ref{globalU1}): indeed, writing
$S=\alpha S_0$ with $S_0\in\SU(2)$ and $|\alpha|=1$ (such an
$\alpha$ exists because $\alpha^2=\det S$), we have
$\Psi S=\alpha\,\Psi S_0$, and the phase $\alpha$ is removed by
(\ref{globalU1}). Thus no separate $\U(2)$ symmetry is required.
Then
\[
u^T S=\sqrt{a}\,e_1^T U^T\bar U R=\sqrt{a}\,e_2^T,\qquad
v^T S=\sqrt{a}\,e_2^T U^T\bar U R=\sqrt{a}\,e_1^T .
\]
Modulo the global $\U(1)$ phase, right multiplication by $S$ therefore brings $\xi$ and $\eta$ to
\[
\xi S=\sqrt{a}\begin{pmatrix}0&0\\ 0&1\end{pmatrix},\qquad
\eta S=\sqrt{a}\begin{pmatrix}1&0\\ 0&0\end{pmatrix},
\]
so that
\[
\Psi S=\sqrt a\,\Psi_{\mathrm{W}},\qquad
\Psi_{\mathrm{W}}
=
\begin{pmatrix}0&0\\0&1\\1&0\\0&0\end{pmatrix}.
\]

\emph{Step 3 (return to the Dirac basis).} The Weyl and Dirac bases are
related by the unitary matrix
\[
U_0=\frac{1}{\sqrt2}\begin{pmatrix}I_2&-I_2\\ I_2&I_2\end{pmatrix},
\qquad
\gamma^\mu_{\mathrm{W}}=U_0\gamma^\mu_{\mathrm{D}}U_0^\dagger ,
\]
so that $\gamma^\mu_{\mathrm{D}}=U_0^\dagger\gamma^\mu_{\mathrm{W}}U_0$.
The Dirac equation (\ref{YM3}) in the Dirac basis reads
\[
i\gamma^\mu_{\mathrm{D}}
\bigl(\partial_\mu\Psi_{\mathrm{D}}+\Psi_{\mathrm{D}}A_\mu\bigr)
-m\Psi_{\mathrm{D}}=0 .
\]
Substituting $\gamma^\mu_{\mathrm{D}}=U_0^\dagger\gamma^\mu_{\mathrm{W}}U_0$
and multiplying by $U_0$ on the left, we obtain
\[
i\gamma^\mu_{\mathrm{W}}
\bigl(\partial_\mu(U_0\Psi_{\mathrm{D}})+(U_0\Psi_{\mathrm{D}})A_\mu\bigr)
-m(U_0\Psi_{\mathrm{D}})=0 ,
\]
where we used that $U_0$ is constant ($\partial_\mu U_0=0$).
Thus the equation has the same form in the Weyl basis, with
\[
\Psi_{\mathrm{W}}:=U_0\Psi_{\mathrm{D}},
\qquad\text{i.e.}\qquad
\Psi_{\mathrm{D}}=U_0^\dagger\Psi_{\mathrm{W}} .
\]
Applying $U_0^\dagger$ to $\sqrt{a}\,\Psi_{\mathrm{W}}$ gives
\[
\sqrt{a}\,U_0^\dagger\Psi_{\mathrm{W}}
=
\frac{\sqrt{a}}{\sqrt2}\begin{pmatrix}I_2\\ \sigma^3\end{pmatrix}.
\]
Writing $\kappa=\sqrt{a/2}>0$, we obtain
\[
\Psi_0^{(L)}(\kappa)=\kappa\begin{pmatrix}I_2\\[1mm]\sigma^3\end{pmatrix}.
\]
This is
class $\mathrm{(L)}$.

\medskip
We have exhausted all possible $c^\nu$: the zero case, the timelike case,
and the lightlike case. Hence there are no other solutions with $J^\nu=0$.
The theorem is proved.
\end{proof}

\subsection{Summary of the classification}
\label{ssecSummary}

The complete classification of constant solutions of the system
(\ref{fir}), (\ref{sec}) is given by Theorems~\ref{thDYM2},
\ref{thR} and \ref{thJzero} and summarized in
Table~\ref{tab4}.

\begin{table}[ht]
\caption{All types of constant solutions of the $\SU(2)$
Yang--Mills--Dirac equations. For each solution $F^2=cI_2$; the column
``$F^2$'' lists the coefficient $c$. The column ``$\Psi$ type'' refers
to the following labels: $(T)$, $(L)$, $(P)$ are the three classes of
Theorem~\ref{thJzero} (all with $J=0$); $(M)$ denotes the massive
branch (\ref{casei}) of Theorem~\ref{thDYM2} with $J\ne0$; $0$ denotes
the trivial solution $\Psi=0$. Under the full group $O(1,3)$,
Case~9 is equivalent to Case~8 with $\lambda=0$ (see
Remark~\ref{remO13redundancy}); under the proper orthochronous subgroup
$SO^+(1,3)$ the two cases are independent.}\label{tab4}%
\begin{ruledtabular}
\begin{tabular}{cccccccccccc}
case & $d_A$ & $x_A$ & $y_A$ & $d_J$ & $x_J$ & $y_J$ & $m$ & $F^2$ & $\Psi$ type & $\rank(\Psi)$ & expl.\ form \\
\hline
1 & 1 & 0 & 0 & 0 & 0 & 0 & $m=0$ & $0$ & $(L)$ & 2 & (\ref{caseiii})\\
2 & 0 & 1 & 2 & 0 & 1 & 0 & $m\neq 0$ & $8m^4$ & $(M)$ & 1 & (\ref{casei})\\
3 & 0 & 0 & 0 & 0 & 0 & 0 & $\forall m$ & $0$ & $0$ & 0 & (\ref{triv1})\\
4 & 1 & 0 & 0 & 0 & 0 & 0 & $\forall m$ & $0$ & $0$ & 0 & (\ref{triv4})\\
5 & 0 & 1 & 0 & 0 & 0 & 0 & $\forall m$ & $0$ & $0$ & 0 & (\ref{triv2})\\
6 & 0 & 0 & 1 & 0 & 0 & 0 & $\forall m$ & $0$ & $0$ & 0 & (\ref{triv3})\\
7 & 0 & 0 & 0 & 0 & 0 & 0 & $m=0$ & $0$ & $(L)$ & 2 & (\ref{tr3})\\
8 & 0 & 0 & 0 & 0 & 0 & 0 & $m=0$ & $0$ & $(T)$ & 2 & (\ref{tr1})\\
9 & 0 & 0 & 0 & 0 & 0 & 0 & $m=0$ & $0$ & $(P)$ & 2 & (\ref{tr2})\\
\end{tabular}
\end{ruledtabular}
\end{table}

The massive branch (case~2) is the only solution with $F^2\ne0$ and
$J\ne0$; there $F^2=8m^4I_2$, so the scalar entering the Yang--Mills
Lagrangian is $\tr(F^2)=16m^4$. For all other solutions $F_{\mu\nu}=0$
identically, hence $F^2=0$ and $\tr(F^2)=0$. The
field is pure gauge, $A_\mu$ being gauge-equivalent to $A=0$ by the
non-constant gauge transformation $S(x)=\exp(x^\mu A_\mu)$. We keep these
solutions in the classification for completeness.

\section{Perturbation theory around constant solutions}
\label{sectNS}

In the previous section of this paper, we have presented the explicit form of all constant solutions of the $\SU(2)$ Yang--Mills--Dirac equations. Let us consider nonconstant solutions of the system of $\SU(2)$ Yang--Mills--Dirac equations
\begin{eqnarray}
&&\partial_\mu A_\nu-\partial_\nu A_\mu-[A_\mu, A_\nu]=:F_{\mu\nu},\\
&&\partial_\mu F^{\mu\nu} -[A_\mu, F^{\mu\nu}]=J^\nu:=i \Psi^\dagger \gamma^0 \gamma^\nu \Psi-\frac{1}{2}\tr(i \Psi^\dagger \gamma^0 \gamma^\nu\Psi)I_2,\label{YM6}\\
&&i\gamma^\mu (\partial_\mu \Psi+\Psi A_\mu)-m \Psi=0,\qquad m\geq 0,\label{YM7}
\end{eqnarray}
in the form of series of perturbation theory near the constant solutions.

Denote the left-hand part of the system (\ref{YM6}) by
\begin{eqnarray}
H^\nu=H^\nu (A):=\partial_\mu(\partial^\mu A^\nu-\partial^\nu A^\mu- [A^\mu,A^\nu])- [A_\mu,\partial^\mu A^\nu-\partial^\nu A^\mu-[A^\mu,A^\nu]]
\end{eqnarray}
and the left-hand part of the equation (\ref{YM7}) by
\begin{eqnarray}
T=T(\Psi, A):=i\gamma^\mu (\partial_\mu \Psi+\Psi A_\mu)-m \Psi.
\end{eqnarray}
Using these notations, we can rewrite (\ref{YM6}), (\ref{YM7}) in the following way
\begin{eqnarray}
H^\nu(A)=J^\nu(\Psi),\qquad T(\Psi, A)=0.\label{HJT}
\end{eqnarray}
Denote the known constant solutions by $A^\mu_{(0)}$ and $\Psi_{(0)}$. Taking the small parameters $\varepsilon <<1$, $\lambda<<1$, we get
\begin{eqnarray}
&&A^\mu=\sum_{k=0}^\infty \varepsilon^k A^\mu_{(k)}=A^\mu_{(0)}+\varepsilon A^\mu_{(1)}+\varepsilon^2 A^\mu_{(2)}+\cdots,\\
&&\Psi=\sum_{k=0}^\infty \lambda^k \Psi_{(k)}=\Psi_{(0)}+\lambda \Psi_{(1)}+\lambda^2 \Psi_{(2)}+\cdots.
\end{eqnarray}
Substituting these expressions into (\ref{HJT}), we get respectively
\begin{eqnarray}
H^\nu=\sum_{k=0}^\infty \varepsilon^k H^\nu_{(k)}=\sum_{k=0}^\infty \lambda^k J^\nu_{(k)}=J^\nu,\qquad T=\sum_{k=0}^\infty\sum_{l=0}^\infty \varepsilon^k \lambda^l T_{(k, l)}=0,
\end{eqnarray}
where
$$H^\nu_{(k)}=H^\nu_{(k)}(A^\mu_{(0)}, \ldots, A^\mu_{(k)}),\quad J^\nu_{(k)}=J^\nu_{(k)}(\Psi_{(0)}, \ldots, \Psi_{(k)}),\quad T_{(k, l)}=T_{(k, l)}(A^\mu_{(0)}, \ldots, A^\mu_{(k)}, \Psi_{(0)}, \ldots, \Psi_{(l)})$$ are some fixed expressions.

The small parameters $\varepsilon$ and $\lambda$ can be connected in different ways, for example, in the way $\lambda =\varepsilon^r$ with some natural number~$r$.

Let us consider the case $\varepsilon=\lambda$ (with $r=1$) further. We use the notation $T_{(k, l)}=T_{(k+l)}$ in this case. We get the equations
\begin{eqnarray}
&&H^\nu_{(k)}(A^\mu_{(0)}, \ldots, A^\mu_{(k)}) = J^\nu_{(k)}(\Psi_{(0)}, \ldots, \Psi_{(k)}),\qquad k=0, 1, 2, \ldots,\\
&&T_{(k)}(A^\mu_{(0)}, \ldots, A^\mu_{(k)}, \Psi_{(0)}, \ldots, \Psi_{(k)})=0,\qquad k=0, 1, 2, \ldots
\end{eqnarray}
The conditions $H^\nu_{(0)}(A^\mu_{(0)})=J^\nu_{(0)}(\Psi_{(0)})$ and $T_{(0)}(A^\mu_{(0)},\Psi_{(0)})=0$ hold automatically because $(A^\mu_{(0)},\Psi_{(0)})$ are constant solutions of the system (\ref{YM6}), (\ref{YM7}).

For the first approximation, we get the system of linear partial differential equations with constant coefficients
\begin{eqnarray}
H^\nu_{(1)}(A^\mu_{(0)}, A^\mu_{(1)})=J^\nu_{(1)}(\Psi_{(0)}, \Psi_{(1)}),\qquad T_{(1)}(A^\mu_{(0)}, A^\mu_{(1)}, \Psi_{(0)}, \Psi_{(1)})=0\label{firstapr}
\end{eqnarray}
for the variables $\Psi_{(1)}$ and $A^\mu_{(1)}$. The explicit form of the system (\ref{firstapr}) is
\begin{eqnarray}
\partial_\mu\partial^\mu A^\nu_{(1)}-\partial_\mu\partial^\nu A^\mu_{(1)}+[A^\nu_{(0)}, \partial_\mu A^\mu_{(1)}]-2[A^\mu_{(0)}, \partial_\mu A^\nu_{(1)}]+[A_{\mu (0)},\partial^\nu A^\mu_{(1)}]+[A_{\mu (0)},[A^\mu_{(0)},A^\nu_{(1)}]]+[A_{\mu (0)},[A^\mu_{(1)}, A^\nu_{(0)}]]\nonumber\\
+[A_{\mu (1)},[A^\mu_{(0)},A^\nu_{(0)}]]=i(\Psi_{(1)}^\dagger \gamma^0 \gamma^\nu \Psi_{(0)}+ \Psi_{(0)}^\dagger \gamma^0 \gamma^\nu \Psi_{(1)})-\frac{1}{2}\tr(i(\Psi_{(1)}^\dagger \gamma^0 \gamma^\nu \Psi_{(0)}+ \Psi_{(0)}^\dagger \gamma^0 \gamma^\nu \Psi_{(1)}))I_2,\label{firstapr2}\\
i\gamma^\mu(\partial_\mu \Psi_{(1)}+\Psi_{(1)} A_{\mu (0)}+\Psi_{(0)} A_{\mu (1)})-m\Psi_{(1)}=0.
\end{eqnarray}
We can take some solution of this system ($\Psi_{(1)}$, $A^\mu_{(1)}$) and substitute it and the constant solution ($\Psi_{(0)}$, $A^\mu_{(0)}$) into
\begin{eqnarray}
H^\nu_{(2)}(A^\mu_{(0)}, A^\mu_{(1)}, A^\mu_{(2)})=J^\nu_{(2)}(\Psi_{(0)},\Psi_{(1)}, \Psi_{(2)}),\qquad T_{(2)}(A^\mu_{(0)}, A^\mu_{(1)}, A^\mu_{(2)}, \Psi_{(0)},\Psi_{(1)}, \Psi_{(2)})=0.
\end{eqnarray}
We get a system of linear partial differential equations with variable coefficients (dependent on $x\in\R^{1,3}$) for the variables $\Psi_{(2)}$ and $A^\mu_{(2)}$. We can take one of solutions of this system and substitute it into the next equation. In the same way, we can get $\Psi_{(k)}$ and $A^\mu_{(k)}$ for any $k=0, 1, 2, \ldots$ This procedure allows us to obtain approximate solutions of the $\SU(2)$ Yang--Mills--Dirac system of equations up to terms of order $k$.

\bigskip

For the solutions 1 and 2 from Table~\ref{tab4} (the massless and massive
branches of Theorem~\ref{thDYM2}), both $A_{(0)}$ and $\Psi_{(0)}$ are
nonzero, so the first-order system (\ref{firstapr2}) contains all the terms
written above. We do not reproduce it here, since its structure is the same
as in the general case. The specific constant values of $A_{(0)}$ and
$\Psi_{(0)}$ can be substituted from (\ref{caseiii}) and (\ref{casei}),
respectively.

For the solutions 3--6 from Table \ref{tab4}, the system for the first approximation (\ref{firstapr2}) takes the form
\begin{eqnarray}
\partial_\mu\partial^\mu A^\nu_{(1)}-\partial_\mu\partial^\nu A^\mu_{(1)}+[A^\nu_{(0)}, \partial_\mu A^\mu_{(1)}]-2[A^\mu_{(0)}, \partial_\mu A^\nu_{(1)}]+[A_{\mu (0)},\partial^\nu A^\mu_{(1)}]\nonumber\\
+[A_{\mu (0)},[A^\mu_{(0)},A^\nu_{(1)}]]+[A_{\mu (0)},[A^\mu_{(1)}, A^\nu_{(0)}]]+[A_{\mu (1)},[A^\mu_{(0)},A^\nu_{(0)}]]=0,\\
i\gamma^\mu(\partial_\mu \Psi_{(1)}+\Psi_{(1)} A_{\mu (0)})-m\Psi_{(1)}=0.
\end{eqnarray}

For the solutions 7--9 from Table~\ref{tab4}, the system for the first
approximation (\ref{firstapr2}) takes the form
\begin{eqnarray}
\partial_\mu\partial^\mu A^\nu_{(1)}-\partial_\mu\partial^\nu A^\mu_{(1)}
=i(\Psi_{(1)}^\dagger \gamma^0 \gamma^\nu \Psi_{(0)}
+ \Psi_{(0)}^\dagger \gamma^0 \gamma^\nu \Psi_{(1)})
-\frac{1}{2}\tr\big(i(\Psi_{(1)}^\dagger \gamma^0 \gamma^\nu \Psi_{(0)}
+ \Psi_{(0)}^\dagger \gamma^0 \gamma^\nu \Psi_{(1)})\big)I_2,\\
i\gamma^\mu(\partial_\mu \Psi_{(1)}+\Psi_{(0)} A_{\mu (1)})=0.
\end{eqnarray}

\bigskip

In Cases~7--9 of Table~\ref{tab4} we have $A_{(0)}=0$ and $m=0$.
These cases correspond to the classes $(L)$, $(T)$, and $(P)$ of
Theorem~\ref{thJzero}, respectively. Setting $\Psi_{(1)}=0$ reduces the
first-order system (\ref{firstapr2}) to
\begin{eqnarray}
\partial_\mu\partial^\mu A^\nu_{(1)}-\partial_\mu\partial^\nu A^\mu_{(1)}=0,
\qquad
\gamma^\mu\Psi_{(0)}A_{\mu(1)}=0.
\label{linSys}
\end{eqnarray}
The first equation is the free Maxwell equation, solved by plane waves
\begin{eqnarray}
A_{\mu(1)}=\varepsilon_\mu e^{ik\cdot x},\qquad
k\cdot x:=k_\alpha x^\alpha,
\end{eqnarray}
with constant $\varepsilon_\mu\in\su(2)$ satisfying
$k^2\varepsilon^\nu=k^\nu(k\cdot\varepsilon)$, where
$k^2:=k_\alpha k^\alpha$ and $k\cdot\varepsilon:=k_\alpha\varepsilon^\alpha$. The second equation is the
linearized Dirac equation with $A_{(0)}=0$. Plane-wave solutions of the
Yang--Mills equations and of the Dirac equation in a plane-wave background
are well known (see, e.g., \cite{waves, MalecRev} and references therein).

We take the plane-wave ansatz
\begin{eqnarray}
A_{0(1)}=A_{3(1)}=0,\qquad
A_{1(1)}=\tau^3 e^{i\omega(t\pm z)},\qquad
A_{2(1)}=\tau^2 e^{i\omega(t\pm z)},
\label{explicitSol}
\end{eqnarray}
with $\omega\in\R$, $t=x^0$, $z=x^3$, and set $\Psi_{(1)}=0$. The wave
vector $k=(\omega,0,0,\pm\omega)$ is null, and $k\cdot\varepsilon=0$. Hence the free Maxwell equation in
(\ref{linSys}) holds. The remaining equation in (\ref{linSys}) reduces to
\begin{eqnarray}
\gamma^1\Psi_{(0)}\tau^3+\gamma^2\Psi_{(0)}\tau^2=0.
\label{diracAnsatz}
\end{eqnarray}
For the three backgrounds, $\Psi_{(0)}$ is the corresponding canonical
form from Theorem~\ref{thJzero}:
\begin{itemize}
\item Case~7, class $(L)$: $\Psi_{(0)}=\kappa(I_2;\sigma^3)$. The lower
block of (\ref{diracAnsatz}) equals
$-\sigma^1\tau^3-\sigma^2\tau^2$, and the upper block equals
$\sigma^1\sigma^3\tau^3+\sigma^2\sigma^3\tau^2$; using
$\sigma^1\sigma^3=-i\sigma^2$ and $\sigma^2\sigma^3=i\sigma^1$, the upper
block is $-\frac{i}{2i}\sigma^2\sigma^2+\frac{i}{2i}\sigma^1\sigma^1
=-\frac12I_2+\frac12I_2=0$, while the lower block vanishes by
(\ref{paulitau}) below.
\item Case~8, class $(T)$: $\Psi_{(0)}=\kappa(I_2;i\lambda I_2)$ with
$\lambda\in\R$. Both blocks of (\ref{diracAnsatz}) are proportional to
$\sigma^1\tau^3+\sigma^2\tau^2$, with coefficients $i\lambda$ (upper) and
$-1$ (lower), hence both vanish by (\ref{paulitau}).
\item Case~9, class $(P)$: $\Psi_{(0)}=\kappa(0;I_2)$. The upper block of
(\ref{diracAnsatz}) equals $\sigma^1\tau^3+\sigma^2\tau^2$, and the lower
block vanishes identically.
\end{itemize}
In all three cases the required identity is
\begin{eqnarray}
\sigma^1\tau^3+\sigma^2\tau^2
=\frac{1}{2i}\bigl(\sigma^1\sigma^2+\sigma^2\sigma^1\bigr)=0,
\label{paulitau}
\end{eqnarray}
which follows from the anticommutation of the Pauli matrices. Thus the
same ansatz (\ref{explicitSol}) gives a first-order plane-wave correction
for each of the three massless backgrounds.

For Case~2 of Table~\ref{tab4} (the massive branch), the background
$A_{(0)}\neq0$ makes the analysis more involved: the commutator terms in
(\ref{firstapr2}) couple the Lie algebra components, and the plane-wave
ansatz does not reduce the system to independent equations. A detailed
analysis of this case, as well as of the second- and higher-order
corrections, is left for future work.

To summarize, the constant solutions of Section~\ref{sectCS} provide a
natural starting point for the perturbation expansion, and in the
massless case the first-order system already admits explicit plane-wave
solutions. These solutions are elementary (free plane waves with a gauge
structure fixed by the Dirac equation), but they illustrate how the
perturbation scheme works in practice.

\section{Conclusions}
\label{sectConcl}

We have presented the complete classification and the explicit form of all
constant solutions of the $\SU(2)$ Yang--Mills--Dirac equations in Minkowski
space $\R^{1,3}$ (Theorems~\ref{thDYM2}, \ref{thR},
and~\ref{thJzero}). The solutions are written in a suitable
coordinate system and gauge, and the classification is summarized in
Table~\ref{tab4}. For each solution we have given explicit formulas for the
spinor $\Psi$, the potential $A$, the current $J$, and the invariant $F^2$
defined in~(\ref{F2def}). For all solutions obtained, $F^2$ is proportional
to $I_2$. Its coefficient is therefore gauge independent, and the scalar
$\tr(F^2)$ enters the Lagrangian of the Yang--Mills field.

For nondegenerate solutions with $\Psi\neq 0$ and $A\neq 0$, our
classification yields a sharp dichotomy: either $m>0$ and $J\neq 0$, or
$m=0$ and $J=0$. In other words, a nonzero constant current requires a
nonzero mass. The massive branch describes a constant non-Abelian field
configuration with nonvanishing $F^2$; the massless branch has $F=0$ and
$J=0$. Note that the classification is performed modulo \emph{constant}
gauge transformations $S\in\SU(2)$ and constant Lorentz transformations. In
particular, all solutions with $F=0$ (the massless branch of
Theorem~\ref{thDYM2} and all solutions of Theorem~\ref{thR}) are pure gauge:
each is gauge-equivalent to $A=0$ by a nonconstant gauge transformation
$S(x)$. Thus they are physically equivalent to the trivial solution, although
they are not related to it by a constant gauge transformation. We keep these
solutions in the classification for completeness.

Nonconstant solutions can be sought as perturbation series around constant
solutions. The first-order equations form a system of linear PDEs with
constant coefficients, which can be studied analytically or numerically. For
the massless backgrounds with $A_{(0)}=0$, we have found explicit
plane-wave solutions of these equations. At the second and higher orders,
the problem reduces to linear PDEs with variable coefficients. These results
provide a starting point for a local classification of all (nonconstant)
solutions of the classical Yang--Mills--Dirac equations and can be used to
describe the physical vacuum and to better understand quantum gauge theory.

Several natural questions remain open. First, the stability of the constant
Yang--Mills--Dirac fields presented here has not been analyzed. For the pure
Yang--Mills case, the instability of constant fields was discussed in
\cite{C2, C3}, and a similar analysis for the Yang--Mills--Dirac system is
an important task for future work. Second, it would be interesting to
generalize our results to the gauge group $\SU(3)$, relevant for quantum
chromodynamics, or to the general case $\SU(N)$. Third, the first-order
system~(\ref{firstapr2}) has been solved explicitly only for the massless
backgrounds with $A_{(0)}=0$. The analysis of the first-order system around
the massive background, as well as the second and higher orders, remains
open. Fourth, a further open problem is the generalization to an arbitrary,
possibly globally hyperbolic, Lorentzian manifold.

\begin{acknowledgements}

This work is supported by the Russian Science Foundation (project 23-71-10028-P), https://rscf.ru/en/project/23-71-10028-P/.

\end{acknowledgements}

\section*{Author Declarations}

\subsection*{Conflict of Interest}
The author has no conflicts to disclose.

\subsection*{Author Contributions}
D. Shirokov: Conceptualization (Lead); Formal analysis (Lead); Investigation (Lead); Methodology (Lead); Writing -- original draft (Lead); Writing -- review \& editing (Lead).

\subsection*{Data Availability Statement}

Data sharing is not applicable to this article as no new data were created or analyzed in this study.

\appendix

\section{The proof of Theorem \ref{thDYM2}}\label{secA}

\textbf{Cases 1, 2, and 3.} Let us consider the first, the second, and the third cases (\ref{cases}) together supposing $a_1\geq 0$ and $a_2\geq 0$. We have
\begin{eqnarray*}
A^0=A_0=A^1=-A_1=\tau^1=\frac{1}{2i}\sigma^3,\qquad
A^2=-A_2=a_1 \tau^2=\frac{a_1}{2i}\sigma^1,\qquad A^3=-A_3=a_2 \tau^3=\frac{a_2}{2i}\sigma^2.
\end{eqnarray*}
Substituting these expressions into (\ref{sec}), we get
\begin{eqnarray*}
\left(\begin{array}{cc} I & 0 \\ 0 & -I \\ \end{array} \right) \left(\begin{array}{c} \Psi_1 \\ \Psi_2 \\ \end{array} \right)\frac{\sigma^3}{2} -\left(\begin{array}{cc} 0 & \sigma^1 \\ -\sigma^1 & 0 \\ \end{array} \right) \left(\begin{array}{c} \Psi_1 \\ \Psi_2 \\ \end{array} \right)\frac{\sigma^3}{2}-\left(\begin{array}{cc} 0 & \sigma^2 \\ -\sigma^2 & 0 \\ \end{array} \right) \left(\begin{array}{c} \Psi_1 \\ \Psi_2 \\ \end{array} \right)\frac{a_1\sigma^1}{2}\\
-\left(\begin{array}{cc} 0 & \sigma^3 \\ -\sigma^3 & 0 \\ \end{array} \right) \left(\begin{array}{c} \Psi_1 \\ \Psi_2 \\ \end{array} \right)\frac{a_2\sigma^2}{2}-m\left(\begin{array}{c} \Psi_1 \\ \Psi_2 \\ \end{array} \right)=\left(\begin{array}{c} 0 \\ 0 \\ \end{array} \right).
\end{eqnarray*}
We obtain the following two equations
\begin{eqnarray}
\Psi_1 \sigma^3-\sigma^1 \Psi_2 \sigma^3-a_1 \sigma^2 \Psi_2 \sigma^1-a_2 \sigma^3 \Psi_2 \sigma^2-2m \Psi_1&=&0,\label{x901}\\
-\Psi_2 \sigma^3+\sigma^1\Psi_1\sigma^3+a_1\sigma^2\Psi_1\sigma^1+a_2\sigma^3\Psi_1\sigma^2-2m\Psi_2&=&0.\label{x902}
\end{eqnarray}
In Case 1, we have $a_1=a_2=0$ and get
\begin{eqnarray}
\Psi_1 \sigma^3-\sigma^1 \Psi_2 \sigma^3-2m \Psi_1&=&0,\label{x01}\\
-\Psi_2 \sigma^3+\sigma^1\Psi_1\sigma^3-2m\Psi_2&=&0.\label{x02}
\end{eqnarray}
Multiplying (\ref{x01}) on the right by $\sigma^3$ and on the left by $\sigma^1$ and using $(\sigma^1)^2=(\sigma^3)^2=I_2$, we get
\begin{eqnarray}
\Psi_2=\sigma^1\Psi_1-2m\sigma^1\Psi_1\sigma^3.\label{x03}
\end{eqnarray}
Substituting (\ref{x03}) into (\ref{x02}) and using $(\sigma^3)^2=I$, we get $m^2 \Psi_1=0$. If $\Psi_1=0$, then using (\ref{x03}), we get $\Psi_2=0$ and $\Psi=0$, i.e. a contradiction. If $m=0$, then we get $\Psi_2=\sigma^1 \Psi_1$. We have $[A_\mu, [A^\mu, A^\nu]]=0$ in this case. Substituting
$$\Psi= \left(\begin{array}{c} \Psi_1 \\ \hline \Psi_2 \\ \end{array} \right)= \left(\begin{array}{c} \Psi_1 \\ \hline \sigma^1 \Psi_1 \\ \end{array} \right)=
\left(
  \begin{array}{cc}
    G & B \\
    C & D \\ \hline
    C & D \\
    G & B \\
  \end{array}
\right)
$$
into (\ref{fir}), we get
\begin{eqnarray*}
J^0=J^1=i\left(\begin{array}{cc} |G|^2+|C|^2-|B|^2-|D|^2 & 2(\bar{G}B+\bar{C}D) \\ 2(\bar{B}G+\bar{D}C) & |B|^2+|D|^2-|G|^2-|C|^2 \\ \end{array} \right),\qquad
J^2=J^3=0.
\end{eqnarray*}
We obtain the solutions (\ref{caseiii}).

In Cases 2 and 3, from (\ref{x901}) and (\ref{x902}), we get
\begin{eqnarray}
\Psi_1 (\sigma^3-2m I)&=&\sigma^1 \Psi_2 \sigma^3+a_1 \sigma^2 \Psi_2 \sigma^1+a_2 \sigma^3 \Psi_2 \sigma^2,\label{x9011}\\
\Psi_2 (\sigma^3+2m I)&=&\sigma^1\Psi_1\sigma^3+a_1\sigma^2\Psi_1\sigma^1+a_2\sigma^3\Psi_1\sigma^2.\label{x9021}
\end{eqnarray}
If $m\neq \frac{1}{2}$, then the matrices $\sigma^3-2m I$ and $\sigma^3+2m I$ are invertible:
$$(\sigma^3-2m I)(\sigma^3+2m I)=(1-4m^2)I.$$
Taking $\Psi_1$ from (\ref{x9011}) and substituting into (\ref{x9021}), we get a homogeneous system of four linear equations for the four variables $K$, $L$, $M$, and $N$. The matrix of this homogeneous system

\begin{eqnarray*}
\left(
    \begin{array}{cc}
      4m^2(-2m-1)+(a_1^2+a_2^2)(-2m+1) & -4ia_1 m  \\
      4ia_1 m & 4m^2(-2m+1)+(a_1^2+a_2^2)(-2m-1)  \\
      2a_1a_2(-2m+1) & -4ia_2 m  \\
      -4ia_2m & 2a_1a_2(2m+1)  \\
    \end{array}\right.
  \\
    \left.\begin{array}{cc}
       2a_1a_2(-2m+1) & 4ia_2 m \\
     4ia_2 m & 2a_1a_2(2m+1) \\
       4m^2(-2m-1)+(a_1^2+a_2^2)(-2m+1) & 4ia_1m \\
     -4ia_1m & 4m^2(-2m+1)+(a_1^2+a_2^2)(-2m-1) \\
    \end{array}
  \right)
\end{eqnarray*}
has determinant
$$(1+2m)^2(1-2m)^2((a_1-a_2)^2+4m^2)^2((a_1+a_2)^2+4m^2)^2.$$
If $m\neq0$ or $a_1\neq a_2$, this determinant is nonzero, hence
$K=L=M=N=0$; we conclude that $\Psi_2=0$, $\Psi_1=0$, and $\Psi=0$,
i.e.\ we get a contradiction, and no solutions of this type exist.

Consider the remaining subcase $m=0$, $a_1=a_2=a>0$. Here the
determinant vanishes, the matrix has rank two, and its kernel is spanned
by $(K,L,M,N)=(K,L,-K,L)$. From (\ref{x9011}) we then get
$$\Psi_1=\left(\begin{array}{cc}-K&L\\K&L\end{array}\right),\qquad
\Psi_2=\left(\begin{array}{cc}K&L\\-K&L\end{array}\right).$$
A direct computation gives $\Psi^\dagger\gamma^0\gamma^2\Psi=0$, hence
$J^2=0$. On the other hand, for the explicit form of $A$ in Case~3 with
$a_1=a_2=a$, the $\nu=2$ component of the left-hand side of (\ref{fir})
equals
$$[A_\mu,[A^\mu,A^2]]=[A_3,[A^3,A^2]]
=[-a\tau^3,[a\tau^3,a\tau^2]]
=a^3[\tau^3,\tau^1]
=a^3\tau^2\neq 0,$$
where the $\mu=0$ and $\mu=1$ terms cancel. This contradicts $J^2=0$, so
no solution exists in this subcase either.

Let us consider the case of $m=\frac{1}{2}$. Substituting
$$\Psi_1=\left(
  \begin{array}{cc}
    G & B \\
    C &D \\
  \end{array}\right),\qquad \Psi_2=\left(
  \begin{array}{cc}
    K & L \\
    M & N \\
  \end{array}
\right)
$$
into (\ref{x9011}) and (\ref{x9021}), we get the system
\begin{eqnarray*}
\left(
  \begin{array}{cc}
    0 & -2B \\
    0 & -2D \\
  \end{array}\right)&=&\left(
  \begin{array}{cc}
    M & -N \\
    K & -L \\
  \end{array}\right)+a_1 \left(
  \begin{array}{cc}
    -Ni & -Mi \\
    Li & Ki \\
  \end{array}\right)+a_2 \left(
  \begin{array}{cc}
    L i & -K i \\
    -N i & M i\\
  \end{array}\right),\\
\left(
  \begin{array}{cc}
    2K & 0 \\
    2M & 0 \\
  \end{array}\right)&=&\left(
  \begin{array}{cc}
    C & -D \\
    G & -B \\
  \end{array}\right)+a_1 \left(
  \begin{array}{cc}
    -Di & -Ci \\
    Bi & Gi \\
  \end{array}\right)+a_2 \left(
  \begin{array}{cc}
B i & -G i \\
    -D i &C i \\
  \end{array}\right).
\end{eqnarray*}
Eliminating $B$, $D$, $K$, and $M$, we get a homogeneous system of four linear equations for the four variables $G$, $C$, $L$, and $N$ with the matrix
$$\left(
  \begin{array}{cccc}
    1-a_1^2-a_2^2 & 2a_1a_2& 2ia_2 & 2ia_1 \\
    2a_1a_2 & 1-a_1^2-a_2^2 & -2ia_1 & -2ia_2 \\
    2ia_2 & -2ia_1 & 1-a_1^2-a_2^2&  -2a_1a_2 \\
    2ia_1 & -2ia_2 & -2a_1a_2 & 1-a_1^2-a_2^2 \\
  \end{array}
\right).$$
The determinant of this matrix is nonzero
$$(1+(a_1-a_2)^2)^2(1+(a_1+a_2)^2)^2>0,$$
thus $G=C=L=N=0$ and then $M=K=D=B=0$. We get a contradiction, we have no solutions in this case.

\textbf{Cases 4, 5, and 6.} Let us consider the fourth, the fifth, and the sixth cases (\ref{cases}) together supposing $a_1>0$, $a_2\geq 0$, and $a_3\geq 0$.

We have
\begin{eqnarray}
A^0=A_0=a_1 \tau^1=\frac{a_1}{2i}\sigma^3,\qquad A^1=-A_1=a_2 \tau^2=\frac{a_2}{2i}\sigma^1,\qquad
 A^2=-A_2=a_3 \tau^3=\frac{a_3}{2i}\sigma^2,\qquad A^3=A_3=0.\label{Ax}
\end{eqnarray}
Substituting these expressions into (\ref{sec}), we get
\begin{eqnarray*}
\left(\begin{array}{cc} I & 0 \\ 0 & -I \\ \end{array} \right) \left(\begin{array}{c} \Psi_1 \\ \Psi_2 \\ \end{array} \right)\frac{a_1\sigma^3}{2} - \left(\begin{array}{cc} 0 & \sigma^1 \\ -\sigma^1 & 0 \\ \end{array} \right) \left(\begin{array}{c} \Psi_1 \\ \Psi_2 \\ \end{array} \right)\frac{a_2\sigma^1}{2}
- \left(\begin{array}{cc} 0 & \sigma^2 \\ -\sigma^2 & 0 \\ \end{array} \right) \left(\begin{array}{c} \Psi_1 \\ \Psi_2 \\ \end{array} \right)\frac{a_3\sigma^2}{2}=m\left(\begin{array}{c} \Psi_1 \\ \Psi_2 \\ \end{array} \right).
\end{eqnarray*}
We obtain the following two equations
\begin{eqnarray}
a_1 \Psi_1 \sigma^3-a_2 \sigma^1\Psi_2 \sigma^1-a_3 \sigma^2 \Psi_2 \sigma^2&=&2m\Psi_1,\label{x1}\\
-a_1 \Psi_2 \sigma^3+a_2 \sigma^1 \Psi_1 \sigma^1+a_3 \sigma^2 \Psi_1 \sigma^2&=&2m \Psi_2.\label{x2}
\end{eqnarray}
In Case 4, we have $a_2=a_3=0$ and get the system
$$a_1 \Psi_1 \sigma^3=2m\Psi_1,\qquad -a_1 \Psi_2 \sigma^3=2m \Psi_2.$$
Thus $$\Psi_1(a_1 \sigma^3-2m I_2)=0,\qquad \Psi_2(a_1 \sigma^3+2m I_2)=0.$$
If $a_1\neq 2m$, then $\Psi_1=\Psi_2=0$, i.e. $\Psi=0$ and we obtain a contradiction. Let us consider the case $a_1=2m>0$. For
$$\Psi=\left(\begin{array}{c} \Psi_1 \\ \hline \Psi_2 \\ \end{array} \right)=\left(
  \begin{array}{cc}
    G & B \\
    C & D \\ \hline
    K & L \\
    M & N \\
  \end{array}
\right),$$
we obtain
$$
\left(\begin{array}{cc}
    G & B \\
    C & D \\
  \end{array}
\right) \left(\begin{array}{cc}
    0 & 0 \\
    0 & -4m \\
  \end{array}
\right)=0,\qquad \left(\begin{array}{cc}
    K & L \\
    M & N \\
  \end{array}
\right) \left(\begin{array}{cc}
    4m & 0 \\
    0 & 0 \\
  \end{array}
\right)=0,
$$
i.e. $B=D=K=M=0$. In this case, we have $[A_\mu, [A^\mu, A^\nu]]=0$.
We calculate the current
\begin{eqnarray*}
&&J^0=\frac{i}{2}\left(\begin{array}{cc}|G|^2+|C|^2-|L|^2-|N|^2 & 0 \\0 & |L|^2+|N|^2-|G|^2-|C|^2 \\ \end{array}\right),\qquad
J^1=i\left(\begin{array}{cc}  0 &  \bar{C}L+\bar{G}N \\ \bar{N}G+\bar{L}C & 0 \\ \end{array}\right),\\
&&J^2=i\left(\begin{array}{cc}  0 & i(\bar{C}L-\bar{G}N) \\  i(\bar{N}G-\bar{L}C) & 0 \\ \end{array}\right),\qquad
J^3=i\left(\begin{array}{cc}  0 & \bar{G}L-\bar{C}N \\  \bar{L}G-\bar{N}C & 0 \\ \end{array}\right),
\end{eqnarray*}
and obtain
\begin{eqnarray*}
|G|^2+|C|^2=|L|^2+|N|^2,\qquad  \bar{G}L-\bar{C}N=0,\qquad \bar{C}L=\bar{G}N=0,
\end{eqnarray*}
i.e. $G=C=L=N=0$. We get $\Psi=0$, i.e. a contradiction.  We have no solutions in this case.

Let us consider Case 5 (we have $a_1>0$, $a_2>0$, $a_3=0$). Multiplying (\ref{x1}) on the right and on the left by $\sigma^1$ and using $(\sigma^1)^2=I$, we get
\begin{eqnarray}
\Psi_2=\frac{a_1}{a_2} \sigma^1\Psi_1 \sigma^3\sigma^1-\frac{2m}{a_2}\sigma^1\Psi_1\sigma^1.\label{x3}
\end{eqnarray}
Substituting this expression into (\ref{x2}), using $(\sigma^1)^2=I$ and $\sigma^1 \sigma^3=-\sigma^3 \sigma^1$, we get
$$\Psi_1((a_1^2+a_2^2+4m^2)I-4a_1 m \sigma^3)=0.$$
The determinant of the matrix $(a_1^2+a_2^2+4m^2)I-4a_1 m \sigma^3$ is nonzero
$$((a_1-2m)^2+a_2^2)((a_1+2m)^2+a_2^2)>0,$$
thus $\Psi_1=0$. Then from (\ref{x3}), we get $\Psi_2=0$, and $\Psi=0$, i.e. we obtain a contradiction.

Now let us consider Case 6 (with $a_1, a_2, a_3>0$). From (\ref{x1}) and (\ref{x2}), we get
\begin{eqnarray}
&&\Psi_1 \left(\begin{array}{cc} a_1-2m & 0 \\ 0 & -2m-a_1 \\ \end{array} \right)=a_2 \sigma^1 \Psi_2 \sigma^1+a_3 \sigma^2 \Psi_2 \sigma^2,\label{x4}\\
&&\Psi_2  \left(\begin{array}{cc} -a_1-2m & 0 \\ 0 & -2m+a_1 \\ \end{array} \right)=-a_2 \sigma^1 \Psi_1 \sigma^1-a_3 \sigma^2 \Psi_1 \sigma^2.\label{x5}
\end{eqnarray}
In the case $a_1\neq 2m$, we can express $\Psi_1$ through $\Psi_2$ from (\ref{x4}) and substitute into (\ref{x5}). Denoting
\begin{eqnarray*}
\Psi_2=\left(\begin{array}{cc} K & L \\ M & N \\ \end{array} \right),\qquad K, L, M, N\in\C,
\end{eqnarray*}
we get the following equations for the variables $K$, $L$, $M$, and $N$:
\begin{eqnarray*}
K(-2m+a_1)((2m+a_1)^2+(a_2+a_3)^2)=0,\qquad
L(2m+a_1)((-2m+a_1)^2+(a_2-a_3)^2)=0,\\
M(-2m+a_1)((2m+a_1)^2+(a_2-a_3)^2)=0,\qquad
N(2m+a_1)((-2m+a_1)^2+(a_2+a_3)^2)=0,
\end{eqnarray*}
i.e. $K=L=M=N=0$. If $\Psi_2=0$, then $\Psi_1=0$ and $\Psi=0$, and we obtain a contradiction.

Let us consider the case $a_1=2m >0$. We get from (\ref{x4}) and (\ref{x5}) the equations
\begin{eqnarray*}
-4m \left(\begin{array}{cc} 0 & B \\ 0 & D \\ \end{array} \right)&=&a_2 \left(\begin{array}{cc} N & M \\ L & K \\ \end{array} \right) +a_3 \left(\begin{array}{cc} N & -M \\ -L & K \\ \end{array} \right),\\
-4m \left(\begin{array}{cc} K & 0 \\ M & 0 \\ \end{array} \right)&=&-a_2 \left(\begin{array}{cc} D & C \\ B & G \\ \end{array} \right) -a_3 \left(\begin{array}{cc} D & -C \\ -B & G \\ \end{array} \right).
\end{eqnarray*}
Solving the eight equations, we obtain $a_2=a_3$ and $G=D=K=N=B=M=0$. Substituting (\ref{Ax}) into (\ref{fir}), we get the following expressions on the left side of the equations
\begin{eqnarray}
&&J^0=[A_1,[A^1, A^0]]+[A_2, [A^2, A^0]]=-4((A^1)^2+(A^2)^2) A^0=\frac{-ia_1(a_2^2+a_3^2)}{2} \sigma^3,\nonumber\\
&&J^1=[A_0,[A^0, A^1]]+[A_2, [A^2, A^1]]=4 ((A^0)^2-(A^2)^2) A^1=\frac{ia_2(a_1^2-a_3^2)}{2}  \sigma^1,\label{J1}\\
&&J^2=[A_0,[A^0, A^2]]+[A_1, [A^1, A^2]]=4 ((A^0)^2 -(A^1)^2) A^2=\frac{ia_3(a_1^2-a_2^2)}{2} \sigma^2,\qquad\qquad J^3=0.\nonumber
\end{eqnarray}
On the right side of (\ref{fir}), we obtain
\begin{eqnarray}
J^0=\frac{i}{2}\left(\begin{array}{cc} |C|^2-|L|^2 & 0 \\ 0 & |L|^2-|C|^2 \\ \end{array} \right),\qquad J^1=i\left(\begin{array}{cc} 0 & \bar{C}L \\ \bar{L}C & 0 \\ \end{array} \right),\qquad
J^2=\left(\begin{array}{cc} 0 & -\bar{C}L \\ \bar{L}C & 0\\ \end{array} \right),\qquad J^3=0.\label{J2}
\end{eqnarray}
Equating (\ref{J1}) and (\ref{J2}), we get $a_1=a_2=a_3=2m$, $C=0$, and the solutions (\ref{casei}).

For completeness, let us compute the invariant $F^2=F_{\mu\nu}F^{\mu\nu}$
for the solution (\ref{casei}). In this case $a_1=a_2=a_3=2m$, and
\begin{eqnarray*}
F_{01}&=&-[A_0,A_1]=-[2m\tau^1,-2m\tau^2]=4m^2\tau^3,\\
F_{02}&=&-[A_0,A_2]=-[2m\tau^1,-2m\tau^3]=-4m^2\tau^2,\\
F_{12}&=&-[A_1,A_2]=-[-2m\tau^2,-2m\tau^3]=-4m^2\tau^1,
\end{eqnarray*}
with all other components vanishing. Raising the indices with
$\eta=\diag(1,-1,-1,-1)$ gives
\[
F^{01}=-4m^2\tau^3,\qquad F^{02}=4m^2\tau^2,\qquad F^{12}=-4m^2\tau^1.
\]
Using $(\tau^a)^2=-\frac14 I_2$ for $a=1,2,3$, we obtain
\begin{eqnarray*}
F^2=F_{\mu\nu}F^{\mu\nu}
&=&2\big(F_{01}F^{01}+F_{02}F^{02}+F_{12}F^{12}\big)\\
&=&2\Big(4m^4I_2+4m^4I_2-4m^4I_2\Big)=8m^4I_2,
\end{eqnarray*}
in agreement with (\ref{casei}).

\textbf{Cases 7, 8, and 9.} Let us consider the seventh, the eighth, and the ninth cases (\ref{cases}) together supposing $a_1>0$, $a_2\geq 0$, and $a_3\geq 0$.

We have
\begin{eqnarray}
A^0=A_0=0,\qquad A^1=-A_1=a_1 \tau^1=\frac{a_1}{2i}\sigma^3,\qquad
A^2=-A_2=a_2 \tau^2=\frac{a_2}{2i}\sigma^1,\qquad A^3=-A_3=a_3 \tau^3=\frac{a_3}{2i}\sigma^2.\label{Ax2}
\end{eqnarray}
Substituting these expressions into (\ref{sec}), we get
\begin{eqnarray*}
-\left(\begin{array}{cc} 0 & \sigma^1 \\ -\sigma^1 & 0 \\ \end{array} \right) \left(\begin{array}{c} \Psi_1 \\ \Psi_2 \\ \end{array} \right)\frac{a_1\sigma^3}{2}- \left(\begin{array}{cc} 0 & \sigma^2 \\ -\sigma^2 & 0 \\ \end{array} \right) \left(\begin{array}{c} \Psi_1 \\ \Psi_2 \\ \end{array} \right)\frac{a_2\sigma^1}{2}
- \left(\begin{array}{cc} 0 & \sigma^3 \\ -\sigma^3 & 0 \\ \end{array} \right) \left(\begin{array}{c} \Psi_1 \\ \Psi_2 \\ \end{array} \right)\frac{a_3\sigma^2}{2}=m\left(\begin{array}{c} \Psi_1 \\ \Psi_2 \\ \end{array} \right).
\end{eqnarray*}
We obtain the following two equations
\begin{eqnarray}
a_1 \sigma^1\Psi_2 \sigma^3+a_2 \sigma^2 \Psi_2 \sigma^1 +a_3 \sigma^3 \Psi_2 \sigma^2&=&-2m\Psi_1,\label{x12}\\
a_1 \sigma^1 \Psi_1 \sigma^3+a_2\sigma^2\Psi_1\sigma^1+a_3 \sigma^3 \Psi_1 \sigma^2&=&2m \Psi_2.\label{x22}
\end{eqnarray}
Let us consider Case 7. We have $a_1>0$, $a_2=a_3=0$, and get the system
$$a_1 \sigma^1\Psi_2 \sigma^3=-2m\Psi_1,\qquad a_1 \sigma^1 \Psi_1 \sigma^3=2m \Psi_2.$$
From the first equation, we get $\Psi_2=-\frac{2m}{a_1}\sigma^1\Psi_1\sigma^3$. Substituting this expression into the second equation, we get
$(a_1^2+4m^2)\Psi_1=0$, i.e. $\Psi_1=0$. Thus $\Psi_2=0$, $\Psi=0$, and we get a contradiction.

Let us consider Cases 8 and 9 (we have $a_1>0$, $a_2>0$, $a_3\geq 0$). If $m\neq 0$, then from the equation (\ref{x12}), we get
$\Psi_1=-\frac{1}{2m}(a_1 \sigma^1\Psi_2 \sigma^3+a_2 \sigma^2 \Psi_2 \sigma^1+a_3 \sigma^3\Psi_2 \sigma^2)$. Substituting this expression into (\ref{x22}), we get
\begin{eqnarray*}
(a_1^2+a_2^2+a_3^2+4m^2)\Psi_2+2a_1 a_2 \sigma^1 \sigma^2\Psi_2 \sigma^1 \sigma^3+2a_1 a_3 \sigma^1 \sigma^3\Psi_2 \sigma^2 \sigma^3
+2a_2 a_3 \sigma^2 \sigma^3\Psi_2 \sigma^2 \sigma^1=0.
\end{eqnarray*}
For $\Psi_2=\left(\begin{array}{cc} K & L \\ M & N \\ \end{array} \right)$, we obtain the equation
\begin{eqnarray*}
(a_1^2+a_2^2+a_3^2+4m^2)\left(\begin{array}{cc} K & L \\ M & N \\ \end{array} \right) +2 a_1 a_2 \left(\begin{array}{cc} iL & -iK \\ -iN & iM \\ \end{array} \right)
+2a_1 a_3 \left(\begin{array}{cc} -iN & -iM \\ iL & iK \\ \end{array} \right)+2a_2 a_3 \left(\begin{array}{cc} M & -N \\ K & -L \\ \end{array} \right)=\left(\begin{array}{cc} 0 & 0 \\ 0 & 0 \\ \end{array} \right),
\end{eqnarray*}
which can be rewritten in the form of a homogeneous system of four linear equations with the matrix
$$
\left(
  \begin{array}{cccc}
    a_1^2+a_2^2+a_3^2+4m^2 & 2ia_1 a_2 & 2a_2 a_3 & -2ia_1 a_3 \\
    -2ia_1 a_2 & a_1^2+a_2^2+a_3^2+4m^2 & -2ia_1 a_3 & -2a_2 a_3 \\
    2a_2 a_3 & 2ia_1 a_3 & a_1^2+a_2^2+a_3^2+4m^2 & -2ia_1 a_2 \\
    2ia_1 a_3 & -2a_2 a_3 & 2ia_1 a_2 & a_1^2+a_2^2+a_3^2+4m^2 \\
  \end{array}
\right).
$$
The determinant of this matrix is nonzero
\begin{eqnarray*}
((a_1+a_2+a_3)^2+4m^2)((a_1-a_2+a_3)^2+4m^2)((a_1+a_2-a_3)^2+4m^2)
((-a_1+a_2+a_3)^2+4m^2)>0,
\end{eqnarray*}
thus $\Psi_2=0$, $\Psi_1=0$ and $\Psi=0$, and we obtain a contradiction. We have no solutions of this type.

If $m=0$, then the system (\ref{x12}), (\ref{x22}) takes the form
\begin{eqnarray}
a_1 \left(\begin{array}{cc} M & -N \\ K & -L \\ \end{array} \right)+a_2 \left(\begin{array}{cc} -iN & -iM \\ iL & iK \\ \end{array} \right)+a_3\left(\begin{array}{cc} iL & -iK \\ -iN & iM \\ \end{array} \right)&=&0,\label{t45}\\
a_1 \left(\begin{array}{cc} C & -D \\ G & -B \\ \end{array} \right)+a_2 \left(\begin{array}{cc} -iD & -iC \\ iB & iG \\ \end{array} \right)+a_3\left(\begin{array}{cc} iB & -iG\\ -iD & iC \\ \end{array} \right)&=&0.\nonumber
\end{eqnarray}
In Case 8 ($a_3=0$), we get
$$\Psi=\left(
  \begin{array}{cc}
    G & iG \\
    iD & D \\
    K & iK \\
    iN & N \\
  \end{array}
\right),\qquad a_1=a_2.$$
Equating the left side of (\ref{fir})
\begin{eqnarray*}
&&J^0=0,\qquad J^1=[A_2,[A^2, A^1]]=-4(A^2)^2 A^1=a_2^2 a_1\tau^1=\frac{a_2^2 a_1}{2i}\sigma^3,\\
&&J^2=[A_1,[A^1, A^2]]=-4(A^1)^2 A^2=a_1^2 a_2 \tau^2=\frac{a_1^2 a_2}{2i}\sigma^1,\qquad J^3=0,
\end{eqnarray*}
and the right side of (\ref{fir})
\begin{eqnarray}
&&J^0=\left(\begin{array}{cc} 0 & -|G|^2+|D|^2-|K|^2+|N|^2 \\ |G|^2-|D|^2+|K|^2-|N|^2 & 0 \\ \end{array} \right),\nonumber\\
&&J^1=i\left(\begin{array}{cc} -i\bar{N}G+i\bar{K}D-i\bar{D}K+i\bar{G}N & \bar{N}G+\bar{K}D+\bar{D}K+\bar{G}N \\ \bar{N}G+\bar{K}D+\bar{D}K+\bar{G}N &  i\bar{N}G-i\bar{K}D+i\bar{D}K-i\bar{G}N \\ \end{array} \right),\nonumber\\
&&J^2=i\left(\begin{array}{cc} \bar{N}G+\bar{K}D+\bar{D}K+\bar{G}N & i\bar{N}G-i\bar{K}D+i\bar{D}K-i\bar{G}N \\ i\bar{N}G-i\bar{K}D+i\bar{D}K-i\bar{G}N & -\bar{N}G-\bar{K}D-\bar{D}K-\bar{G}N\\ \end{array} \right),\nonumber\\
&&J^3=i\left(\begin{array}{cc} 0 & i(\bar{K}G+\bar{N}D+\bar{G}K+\bar{D}N) \\ -i(\bar{K}G+\bar{N}D+\bar{G}K+\bar{D}N) & 0\\ \end{array} \right),\nonumber
\end{eqnarray}
we get
\begin{eqnarray*}
\bar{N}G-\bar{K}D+\bar{D}K-\bar{G}N=\frac{a_2^2 a_1}{2i}=-\frac{a_1^2 a_2}{2i},
\end{eqnarray*}
which is impossible for 
$a_1, a_2>0$. We have no solutions in this case.

In Case 9 ($a_3\neq 0$), we get from (\ref{t45})
\begin{eqnarray*}
L=\frac{i}{a_1}(a_2 K+a_3 M),\qquad N=\frac{-i}{a_1}(a_2 M+a_3 K),\qquad
K=\frac{a_1^2-a_2^2-a_3^2}{2a_2a_3}M,\qquad (a_1^2-(a_2+a_3)^2)(a_1^2-(a_2-a_3)^2)M=0.
\end{eqnarray*}
If $M\neq 0$, then $a_1=a_2+a_3$, because we can always take the hyperbolic singular values in decreasing order $a_1>a_2>a_3>0$. If $M=0$, then $K=0$, and $L=N=0$, thus $\Psi_2=0$. We get the same for the variables $G, B, C, D$ from (\ref{t45}).
We obtain
$$\Psi=\left(
  \begin{array}{cc}
    G & \frac{i}{a_2+a_3}(a_2 G+ a_3 C) \\
    C & \frac{-i}{a_2+a_3}(a_2 C+ a_3 G) \\
    K & \frac{i}{a_2+a_3}(a_2 K+ a_3 M) \\
    M & \frac{-i}{a_2+a_3}(a_2 M+ a_3 K) \\
  \end{array}
\right).
$$
On the left side of (\ref{fir}), we get
\begin{eqnarray}
&&J^0=0,\qquad J^1=[A_2,[A^2, A^1]]+[A_3,[A^3,A^1]]=(a_2^2+a_3^2)a_1 \tau^1,\label{t56}\\
&&J^2=[A_1,[A^1, A^2]]+[A_3,[A^3,A^2]]=(a_1^2+a_3^2)a_2 \tau^2,\qquad J^3=[A_1,[A^1, A^3]]+[A_2,[A^2,A^3]]=(a_1^2+a_2^2)a_3 \tau^3.\nonumber
\end{eqnarray}
Calculating $J^0$ on the right side of (\ref{fir}), we get
\begin{eqnarray*}
&&(a_2+a_3)^2(|G|^2+|C|^2+|K|^2+|M|^2)=(|a_2 G+a_3 C|^2+|a_2 C+a_3 G|^2+|a_2 K+a_3 M|^2+|a_2 M+a_3 K|^2),\\
&&a_2(|G|^2-|C|^2+|K|^2-|M|^2)+a_3(\bar{G}C-\bar{C}G+\bar{K}M-\bar{M}K)=0.
\end{eqnarray*}
Simplifying the first equation, we get $|G-C|^2+|K-M|^2=0$, i.e. $G=C$ and $K=M$. We conclude that
$$\Psi=\left(
  \begin{array}{cc}
    G & iG \\
    G & -iG \\
    K & iK \\
    M & -iK \\
  \end{array}
\right).
$$
Equating
\begin{eqnarray*}
&&J^1=i\left(\begin{array}{cc} 2(\bar{K}G+\bar{G}K) & 0 \\ 0 &  -2(\bar{K}G+\bar{G}K) \\ \end{array} \right),\qquad
J^2=i\left(\begin{array}{cc} 0 & -2(\bar{K}G+\bar{G}K) \\ -2(\bar{K}G+\bar{G}K) & 0 \\ \end{array} \right),\\
&&J^3=i\left(\begin{array}{cc} 0 & 2i(\bar{K}G+\bar{G}K) \\  -2i(\bar{K}G+\bar{G}K) & 0\\ \end{array} \right),
\end{eqnarray*}
and (\ref{t56}), we get
$$4(\bar{K}G+\bar{G}K)=-a_1(a_2^2+a_3^2)=a_2(a_1^2+a_3^2)=a_3(a_1^2+a_2^2),$$
i.e. a contradiction. We have no solutions in this case.

The theorem is proved.

\bibliography{aipsamp}
\end{document}